\def\draft{0}  
    \def\ITCS{0} 
    \def\EC{0} 

\ifnum\EC=0

    \documentclass[11pt]{article}
    \usepackage[letterpaper,hmargin=1in,vmargin=1in]{geometry}
    \ifnum\ITCS=0
    \fi

    \usepackage{verbatim,sgame}
\fi

    \usepackage{natbib,nicefrac}
    
    \usepackage{tikz}
    \usetikzlibrary{positioning,chains,fit,shapes,calc}
    \usepackage{float,subcaption}
    \floatstyle{boxed}
    \restylefloat{figure, graphicx}

    \usepackage{amsfonts,url,ifthen,epsfig,amsmath,amssymb,color,framed,float,mathrsfs}
    \usepackage[colorlinks=true,breaklinks=true,bookmarks=true,urlcolor=black,
citecolor=black,linkcolor=black,bookmarksopen=false,draft=false]{hyperref}
\ifnum\EC=0
    \usepackage{algorithm}
        \floatname{algorithm}{Mediator}

\fi
    \usepackage[noend]{algpseudocode}
    \usepackage{etoolbox}
    
    \makeatletter
    \let\original@footnotemark\footnotemark
    \newcommand{\align@footnotemark}{%
      \ifmeasuring@
        \chardef\@tempfn=\value{footnote}%
        \original@footnotemark
        \setcounter{footnote}{\@tempfn}%
      \else
        \iffirstchoice@
          \original@footnotemark
        \fi
      \fi}
    \pretocmd{\start@align}{\let\footnotemark\align@footnotemark}{}{}
    \makeatother

    \ifnum\draft=1
    \newcommand{\Rnote}[1]{\begin{framed}\noindent \textcolor{red}{{#1}}\end{framed}} 
    \else
    \newcommand{\Rnote}[1]{}
    \fi

    \newcommand{\remove}[1]{}
    \newtheorem{theorem}{Theorem}
    \newtheorem{example}{Example}

    \newtheorem{claim}{Claim}
    
    \newtheorem{lemma}{Lemma}
    \newtheorem{corollary}{Corollary}
    \newtheorem{obs}{Observation}
    \newtheorem{definition}{Definition}
    \newtheorem{proposition}{Proposition}
    \newtheorem{remk}[theorem]{Remark}
    
    \newenvironment{observation}{\begin{obs} \begin{normalfont}}{\end{normalfont}
    \end{obs}}

    \def\FullBox{\hbox{\vrule width 8pt height 8pt depth 0pt}}
    
    \def\qed{\ifmmode\qquad\FullBox\else{\unskip\nobreak\hfil
    \penalty50\hskip1em\null\nobreak\hfil\FullBox
    \parfillskip=0pt\finalhyphendemerits=0\endgraf}\fi}
    
    \def\qedsketch{\ifmmode\Box\else{\unskip\nobreak\hfil
    \penalty50\hskip1em\null\nobreak\hfil$\Box$
    \parfillskip=0pt\finalhyphendemerits=0\endgraf}\fi}
    
 \ifnum\EC=0    
    \newenvironment{proof}{\begin{trivlist} \item {\bf Proof:~~}}
      {\qed\end{trivlist}}
  \fi

    \newcommand{\beq}{\begin{equation}}
    \newcommand{\eeq}{\end{equation}}
    \newcommand{\be}{\begin{enumerate}}
    \newcommand{\ee}{\end{enumerate}}
    \newcommand{\bi}{\begin{itemize}}
    \newcommand{\ei}{\end{itemize}}
    \newcommand{\bd}{\begin{description}}
    \newcommand{\ed}{\end{description}}
    \newcommand{\bc}{\begin{center}}
    \newcommand{\ec}{\end{center}}
    \newcommand{\bthm}{\begin{theorem}}
    \newcommand{\ethm}{\end{theorem}}
    \newcommand{\bdefi}{\begin{definition}}
    \newcommand{\edefi}{\end{definition}}
    \newcommand{\bcor}{\begin{corollary}}
    \newcommand{\ecor}{\end{corollary}}
    \newcommand{\blem}{\begin{lemma}}
    \newcommand{\elem}{\end{lemma}}
    \newcommand{\bexa}{\begin{example}}
    \newcommand{\eexa}{\end{example}}
    \newcommand{\bprop}{\begin{proposition}}
    \newcommand{\eprop}{\end{proposition}}

    \newcommand{\opt}{\mathrm{opt}}
    
        \newcommand{\oS}{\overline{S}}
    \newcommand{\og}{\overline{g}}

    \ifnum\draft=1
    \newcommand{\ronote}[1]{\begin{framed}\noindent \textcolor{red}{{Ronen's note: #1}}\end{framed}} 
    \newcommand{\ranote}[1]{\begin{framed}\noindent \textcolor{red}{{Rann's note: #1}}\end{framed}} 
    \else
    \newcommand{\ranote}[1]{}
    \newcommand{\ronote}[1]{}
    \fi

    \def\real{\hbox{\rm\setbox1=\hbox{I}\copy1\kern-.45\wd1 R}}
    \def\neal{\hbox{\rm\setbox1=\hbox{I}\copy1\kern-.45\wd1 N}}

    \newcommand{\abs}[1]{\left|{#1}\right|}
    \newcommand{\eps}{\varepsilon}
    
    \newcommand{\pr}[1]{\Pr\left[#1\right]}

    \newcommand{\lp}{\mathrm{LP}}

    \def\OPT{\text{\rm \,OPT}}
    \newcommand{\R}{{\mathbb R}}
    \newcommand{\supp}{\mathrm{supp}}

\newcommand{\Med}{\mathcal{M}}

\ifnum\EC=1

\title[Coopetition Against an Amazon]{Coopetition Against an Amazon}

\author{Submission 620}

\begin{abstract}This paper studies cooperative data-sharing between competitors vying to predict a consumer's tastes. We design optimal data-sharing schemes
       both for when they compete only with each other, and for when they additionally compete with an Amazon---a company with more, better data.
       We derive conditions under which competitors benefit from coopetition, and under which full data-sharing is optimal.
\end{abstract}

\begin{document}
\maketitle

\fi

\ifnum\EC=0
\begin{document}
    
    \definecolor{myblue}{RGB}{80,80,160}
    \definecolor{mygreen}{RGB}{80,160,80}
    
    \title{Coopetition Against an Amazon}
 \ifnum\ITCS=1
 \author{}
 \else   
    \author{Ronen Gradwohl\thanks{Department of Economics and Business Administration, Ariel University. Email: \texttt{roneng@ariel.ac.il}.
    Gradwohl gratefully acknowledges the support of National Science Foundation award number 1718670.
    }
    \and
    Moshe Tennenholtz\thanks{Faculty of Industrial Engineering and Management, The Technion -- Israel Institute of
    Technology. Email: \texttt{moshet@ie.technion.ac.il.} The work by Moshe Tennenholtz was supported by funding from the European
Research Council (ERC) under the European Union's Horizon 2020
research and innovation programme (grant number 740435).}}
    \fi
    \date{}
    
	\maketitle

\begin{abstract}
This paper studies cooperative data-sharing between competitors vying to predict a consumer's tastes. We design optimal data-sharing schemes both for when they compete only with each other, and for when they additionally compete with an Amazon---a company with more, better data. We show that simple schemes---threshold rules that probabilistically induce either full data-sharing between competitors, or the full transfer of data from one competitor to another---are either optimal or approximately optimal, depending on properties of the information structure. We also provide conditions under which firms  share more data when they face stronger outside competition, and describe situations in which this conclusion is reversed.
\end{abstract}

 \section{Introduction}
 A key challenge to firms competing in today's electronic marketplace is competition against Big Tech companies that have considerably more data and so better predictive models.
One way in which smaller firms may overcome this hurdle and survive or even thrive in such a market is to engage in coopetitive
 strategies---namely, to cooperate with other small firms that are its competitors.
 Such coopetitive strategies have increasingly become a field of study by both academics and practitioners (see, for example, \cite{brandenburger2011co} and \cite{bengtsson2000coopetition}), but they largely focus on industrial applications such as healthcare, IT, and service industries.
In this paper, we study coopetition between e-commerce companies, and focus on the possibility of data 
sharing as a way to deal with their data imbalance vis-\`a-vis Big Tech.

Such coopetitive data-sharing can be undertaken by the firms themselves or by external service-providers. To facilitate the former, 
the newly burgeoning area of federated machine learning  has as its goal
the design of mechanisms that generate predictive models for firms based on such
shared data \citep{yang2019federated}. The assumption underlying federated learning is that these firms 
actually want predictive models based on all shared data, 
but in 
competitive scenarios this need not be the case. Although
each firm certainly desires its own model to be as predictive as possible, it most likely does not wish the same
for its competitors. So when is full data-sharing, leading to maximally predictive models for {\em all} firms,
optimal? When is partial or no sharing better? And do these answers differ when firms compete against
an Amazon?

In addition to using federated machine learning, coopetitive data-sharing can also be facilitated by an 
external service-provider.
While tools for data sharing are prevalent---and include platforms such as Google Merchant and Azure Data Share---they
are currently undergoing even further development. For example, in its European Strategy for Data, the European Commission plans the following: ``In the period 2021-2027, the Commission will invest in a High Impact Project on European data spaces and federated cloud infrastructures. The project will fund infrastructures, data-sharing tools, architectures and governance mechanisms for thriving data-sharing and Artificial Intelligence ecosystems'' \citep{EC2020}.
But just as with federated learning, the potential benefits from data sharing may arise under partial, rather than full, sharing.
For instance, one industry white paper urges service providers to
\begin{quote}``...analyze the combined, crowdsourced data and generate benchmark analyses and comparative performance reports. Each participating client gains insights that it could not otherwise access, and each benefits from the service-host provider's ability to {\em slice and dice} the aggregated data and share the results that are relevant to each client'' \citep{loshen2014leveraging}. 
\end{quote}
What is the optimal way for such service providers to combine, slice and dice, 
and share e-commerce firms' data 
in order to facilitate successful coopetition against  an Amazon?

%

Regardless of how coopetitive data-sharing is achieved---via federated learning or through a service provider---the
answer to whether or not firms desire maximally predictive models for all depends on how these predictions are
eventually used. If the firms involved act in completely unrelated markets, 
then of course full data-sharing is optimal---it
yields a maximally predictive model for a firm, and that firm is not harmed by others' better predictions.
This is no longer true if the firms are competitors, since then a firm may be
harmed by an improvement in its competitors' predictions. In this case, whether or not the tradeoff is worthwhile 
depends on the details of the market in which firms compete.

%


In this paper we focus on one such market, which is motivated by recommender systems for online advertising and Long Tail retail.
In our model, there are many types of consumers, and firms use their data to infer 
and take a tailored action for each consumer type. 
In the online advertising market, for example, firms use their data
to personalize advertisements in order to maximize the probability that a consumer clicks on their ad. 
Data sharing is common in this market, and platforms such as Google Merchant and Azure Data Share 
specifically facilitate such sharing across clients in order to improve personalization and increase clickthrough rates.

For another concrete application, consider retailers in a Long Tail market---a market with an enormous number of low-demand goods that collectively make up substantial
market share---interested in finding consumers for their products. Success in such markets crucially relies 
on retailers' abilities to predict consumers' tastes in order to match them to relevant products  \citep{anderson2006long,pathak2010empirical}.
On their own, smaller firms can be destroyed by giants with vastly greater amounts of data and thus a clear predictive advantage.
In cooperation with other small firms, however, this data imbalance may be mitigated, possibly granting smaller firms  a fighting 
chance. (We do not discuss here antitrust issues surrounding data sharing; the reader is referred to \cite{martens2020business} for a discussion thereof in the context of the EU's European Strategy for Data.)

In this paper we study such cooperation between competing firms, and design data-sharing schemes that are optimal---ones
that maximize total firm profits subject to each firm having the proper incentives to participate.
More specifically, we study two settings, one without and the other with an Amazon. 
In the former, in which smaller firms compete only with one another,
 we show that coopetition is beneficial to the extent
that firms can share data to simultaneously improve their respective predictions. In the optimal data-sharing scheme 
firms benefit from sharing because it allows them to better predict the tastes of consumers for whom no firm predicts well on its own.

In the latter setting, in which smaller firms compete with one another but also with an Amazon, 
sharing data about consumers for whom no firm predicts well can be beneficial, but is not optimal. 
We show that, in the optimal scheme, firms share data with others and weaken their own 
market positions on some  market segments, in exchange for receiving data from others and strengthening their market positions on other segments.

The following simple example illustrates our model and previews some of our results. The example is framed
as retailers offering goods to consumers who may make a purchase, but can also be interpreted as 
advertisers displaying advertisements to consumers who may click on the ads.
In either case, suppose the market consists of two goods, $g_0$ and $g_1$, that consumers may desire. Each consumer is described by a feature vector in $\{0,1\}\times\{0,1\}$ that
determines that consumer's taste: Consumers of types 00 and 11 are only interested in $g_0$, whereas consumers of types 01 and 10 are only interested in $g_1$. %
A priori, suppose the distribution of consumers in the market is such that an $\alpha$ fraction are of type 00, a $\beta$ fraction 
of type 01, a $\gamma$ fraction of type 10, and a $\delta$ fraction of type 11.

When a consumer shows up to the market, a retailer may offer him one of the goods (as a convention, we use masculine pronouns for consumers and feminine
pronouns for retailers/players). 
If the retailer has data about the consumer's type, she will offer the good
desired by that consumer. Assume that when a consumer is offered the correct good he makes the purchase, and that this leads to a profit of 1
to the retailer. If there are two retailers offering this good to the consumer he chooses one at random from whom to make the purchase. A retailer whose
good is not chosen gets profit 0.

Suppose now that there are two retailers pursuing the same consumer, but that they do not fully know the consumer's type. 
Instead, the first retailer only knows the first bit of the consumer's type,
and the second retailer only knows the second bit. This situation is summarized in Figure~\ref{fig:2by2}, where the first retailer knows the row bit
and the second knows the column bit.
Which goods should the competing retailers offer the consumer?
For this example, suppose for simplicity that $\alpha>2\beta\geq 2\gamma>4\delta$. Then each of the retailers has a dominant strategy: 
If the bit they know is 0, offer $g_0$, and if the bit they know is 1, offer $g_1$.

To see that this is a dominant strategy, consider for example the row retailer, and suppose she learns that a consumer's row bit is 0.
She thus knows that the correct good is $g_0$ with probability $\alpha/(\alpha+\beta)$ and $g_1$ with probability $\beta/(\alpha+\beta)$. 
Her utility, however, depends also on the good offered by the column retailer. The row retailer's worst-case utility from offering $g_0$ is $\alpha/(2\alpha+2\beta)$, which
occurs when the column player also offers $g_0$. Her best case utility from offering $g_1$ is $\beta/(\alpha+\beta)$. Since we assumed $\alpha>2\beta$,
offering $g_0$ is  best for the row player regardless of the other retailer's offer, and is thus dominant. The analyses for the case in which the row retailer's bit is 1, 
as well as for the column retailer's strategy,
are similar.

 \begin{figure}[tbp]
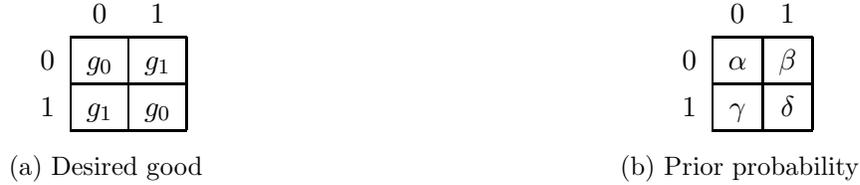

 \centering
\begin{subfigure}[b]{.5\textwidth}
\centering
      \begin{tabular}{ r|c|c| }
\multicolumn{1}{r}{}
 &  \multicolumn{1}{c}{0}
 & \multicolumn{1}{c}{1} \\
\cline{2-3}
0 & $g_0$ & $g_1$ \\
\cline{2-3}
1 & $g_1$ & $g_0$ \\
\cline{2-3}
\end{tabular}
      \caption{Desired good}
        \label{fig:2by2a}
    \end{subfigure}~
\begin{subfigure}[b]{.5\textwidth}
\centering
    \begin{tabular}{ r|c|c| }
\multicolumn{1}{r}{}
 &  \multicolumn{1}{c}{0}
 & \multicolumn{1}{c}{1} \\
\cline{2-3}
0 & $\alpha$ & $\beta$ \\
\cline{2-3}
1 & $\gamma$ & $\delta$ \\
\cline{2-3}
\end{tabular}
        \caption{Prior probability}
        \label{fig:2by2b}
    \end{subfigure}
~ 
\caption{Data sharing: An example}
\label{fig:2by2}
\end{figure}

The result from deployment of these dominant strategies is that consumers of type 00 will be offered the correct good by both retailers, and will thus choose one at random. Consumers
of types 01 and 10 will be offered the correct good by only one of the retailers, whereas consumers of type 11 will not make a purchase as they will not be
offered the correct good. The  
expected profit of the first retailer will thus be $\alpha/2 + \gamma$, and of the second retailer $\alpha/2+\beta$.

Does data sharing improve retailers' profits? Suppose retailers could share their respective information with one another, so that both 
always knew the consumer's type. This would lead to both always offering the correct good, and hence to expected profits of $1/2$ for each. However,
this would be detrimental to the second retailer whenever $\alpha/2+\beta > 1/2$, since her profits with data sharing would be
lower than without. Such cooperation is thus not
{\em individually rational} -- the second retailer will not want to cooperate with the first.

Instead, suppose there is a {\em mediator}---a trusted third-party that
may represent either a service provider or a cryptographic protocol (such as a federated learning algorithm) run
by the two retailers. The European Commission  calls mediators {\em common data spaces} \cite{EC2020}, and describes them as ``third-parties [who] may act as intermediaries and apply new technologies and ways of organising markets in order to...enable transactions that were previously not feasible. They can be private, public or community organisations that are neutral with respect to data uses...'' \citep{martens2020business}.
Given such a mediator, consider the following data-sharing scheme: both retailers share their data with the mediator,
who then passes along the data to all retailers {\em only} if the consumer's type is 11. If the consumer is not of type 11 then each retailer has her original data,
and additionally infers that the consumer's type is not 11 (for otherwise she would have learned that the type is 11). 

How does this scheme affect retailers' strategies? Clearly, when retailers learn that the consumer's type is 11, both offer $g_0$. What
happens if a retailer learns that the consumer's type is not 11? Consider the row retailer. If her bit is 0, then learning that the consumer's type
is not 11 does not provide any new information, and so she still has the same dominant strategy of offering $g_0$. If her bit is 1, however, then she learns that
the consumer's type must be 10. But note that in this case, her best strategy is to offer $g_1$, which is the same as her dominant strategy absent a
mediator. Thus, the mediator changes the row retailer's behavior {\em only} when the consumer is of type 11. A similar analysis and conclusion hold for
the column retailer.

Since retailers  offer the same goods with and without data sharing in all cases except when the consumer is of type 11, this scheme changes retailer's
profits only in this latter case. In particular, it leads
to an additional (ex ante) profit of $\delta/2$ for each retailer beyond her original profit---with probability $\delta$ the consumer is of type 11, in which case the retailers
split the additional surplus of 1---and is thus beneficial to both. Observation~\ref{obs:noA-opt}
shows that such a scheme is, in fact, not only individually rational but also optimal.

Suppose now that the two retailers are also competing against an Amazon for the consumer. Since the Amazon is a giant it has more data,
and, in particular, we assume that it has complete information about each consumer's type. In this three-way competition, the original dominant
strategies of our two smaller retailers do not perform as well: They lead to profits of $\alpha/3+\gamma/2$ and $\alpha/3+\beta/2$ to the
first and second retailer, respectively, since when they offer the correct good the consumer now chooses amongst up to 3 retailers.
The mediator described above, which reveals information when the type is 11, leads to higher profits, since now there is an additional
$\delta/3$ to each of the smaller retailers. However, that mediator is no longer optimal, and the retailers can actually do better.

To see this, observe that, conditional on consumer type 01, the total profit to the small retailers is $\beta/2$, since
only the second makes the correct offer $g_1$ and then competes with the Amazon. In contrast, if both retailers were to know the type and offer the correct good,
then their total profit would be $2\beta/3$, namely $\beta/3$ each. This is harmful to the second retailer as it involves a loss of profit, but a gain for the first
and the sum. But the second retailer can be compensated by getting data from the first elsewhere, for instance on consumer type 10.

For simplicity of this example, suppose that $\beta=\gamma$, and consider a mediator that facilitates full data-sharing, in
which both retailers learn each other's data and thus have complete information about consumers' types. Here, the profit of each is 
$1/3>\alpha/3+\beta/2=\alpha/3+\gamma/2$, so both gain from this data sharing. Furthermore, the total profit of the retailers is $2/3$,
which is the maximal utility they can obtain when competing against an Amazon. Thus, this data-sharing scheme is individually rational and optimal.
Observe that it leads to higher retailer welfare than the optimal scheme absent a mediator.

Full data-sharing is not always individually rational, however. If $\beta$ is much greater than $\gamma$, then the second retailer is not sufficiently
compensated by the first for sharing data about consumer 01. The second retailer will consequently be harmed by such data sharing, and so the scheme
will not be individually rational. However, in
Theorem~\ref{thm:with-amazon} we show that there is a different data-sharing scheme that is individually rational and optimal, 
a scheme in which the first retailer shares all her data
and the second shares some of her data. Overall, we show that such data-sharing coopetition against an Amazon is beneficial to the small retailers.
 
 
 
\subsection{Contribution} Our model is more general than the simple example above, and involves many goods and many types of consumers. For the example it was useful
 to think of each consumer as having a feature vector describing his type, but in the model we take a more general approach that allows for a wider
 class of information structures. Furthermore, in general the retailers will not have dominant strategies, either with or
 without a mediator. Instead, we suppose that, absent a mediator, players play an arbitrary Bayesian Nash equilibrium, and construct mediators that lead to higher expected utilities in equilibrium.


In Section~\ref{sec:baseline},  we design
optimal mediators for coopetitive data-sharing. 
 In addition to maximizing total retailer profits, all our mediators are individually rational---retailers attain higher utilities when they
 use the mediators---and incentive compatible---in equilibrium,
 each retailer's best strategy is to follow the mediators' recommendations. One of our main results is that the optimal
 schemes are simple. In particular, instead of 
 relying on the details of retailers' data or on their strategies in equilibrium sans data-sharing (as in the introductory example above),
our mediators consist of threshold rules that probabilistically induce either full data-sharing between retailers, or the full transfer of data from one retailer to another.

The analysis in Section~\ref{sec:baseline} retains one of the assumptions present in the 
example above---that the joint data of the small retailers
is sufficient to uniquely identify each consumer's type. In Section~\ref{sec:general} we drop this assumption, 
leading to a new set of  challenges. This is because, when the assumption holds, there is no 
conflict between welfare maximization and equilibrium,
and so the main hurdle in designing optimal mediators is the individual rationality constraint. Without the assumption, however,
participants' incentive compatibility constraints impose a limit on the total welfare that can be achieved. 
Nonetheless, our main results
in this section are that variants of the simple mediators from Section~\ref{sec:baseline} are approximately optimal
here as well. We also provide examples showing that our approximation factor is tight.

Finally, in Section~\ref{sec:share-more} we delve into the intriguing question of whether players 
share more data in the presence of
an Amazon or in its absence. First, under the assumption that  the joint data of the small retailers
is sufficient to uniquely identify each consumer's type, we show that if data sharing is strictly beneficial in
the absence of an Amazon, then it is also strictly beneficial in its presence. This confirms the intuition that players  share more data when facing stronger outside competition. However, we also show that this conclusion may be reversed when the assumption does not hold. In this latter case, we show that sometimes data sharing can be strictly beneficial in
the absence of an Amazon, but {\em not} in its presence. We show that the reason this may happen depends on whether or not
players' equilibrium considerations conflict with welfare maximization.

 
The rest of the paper proceeds as follows. First we survey the related literature, and then, in Section~\ref{sec:model}, develop the formal model. This is followed by our main analyses in Sections~\ref{sec:baseline},~\ref{sec:general}, and~\ref{sec:share-more}. Concluding notes appear in Section~\ref{sec:conclusion}.
    
\subsection{Related Literature}
Most broadly, our paper contributes to a burgeoning literature on competition in prediction and machine learning.
Within this literature, papers such as \cite{mansour2018competing}, \cite{ben2019regression}, and \cite{feng2019bias} study 
different models of learning embedded in settings where participants compete with others in their predictions.
They take participants' data as given, and point to the effect competition has on optimal learning algorithms. 
Our focus, in contrast, is on the effect of data sharing on competitive prediction.

Because of its focus on data sharing, our paper is also related to the more established literature on strategic information-sharing, 
a literature that focuses on a number of distinct applications: oligopolistic competition \citep{clarke1983collusion,raith1996general}, financial intermediation \citep{pagano1993information,jappelli2002information,gehrig2007information}, supply chain management \citep{ha2008contracting,shamir2016public}, price discrimination \citep{liu2006customer,jentzsch2013targeted}, and
competition between data brokers \citep{gu2019data,ichihashi2020competing}.
Much of this literature revolves around the question of whether it is beneficial for participating firms to pool all their data, or whether
they would prefer to pool only some or none at all.
As \cite{bergemann2013robust} demonstrate, however, more finely tuned data sharing can be beneficial even when full or partial pooling is not, 
and so the existing schemes do not exhaust the potential benefits of data sharing. In this paper we analyze precisely such
finely tuned sharing.

A different paper that does consider more finely tuned sharing is that of \cite{de2020crowdsourcing}, who study an 
infinitely repeated setting
in which firms compete for market share. Their main result is the construction of a data sharing scheme that Pareto
improves all participants' welfare. Like our results, their scheme also only reveals partial information to 
participants. Unlike our approach, however, \cite{de2020crowdsourcing}'s scheme relies on the repeated nature
of the interaction, and uses folk-theorem-type arguments to show that cooperation can be sustained. In addition,
their focus is on data sharing between many small firms, and they do not consider the possible presence of an Amazon.

In terms of modeling and techniques, our paper falls into the literature on information design. Information design, recently surveyed by \ifnum\EC=1 \cite{dughmi2017algorithmic} and by \fi \cite{bergemann2019information}, is the study of how the allocation of information affects incentives and hence behavior, with a focus on the extent and limits of purely informational manipulation. 
%
The information design problem encompasses work on communication in games and on Bayesian persuasion.
When the mediator is assumed to have only the information held by the players, or only the information they are willing to share with him, 
the problem maps to the one studied in the literature on communication in games. The goal of these studies is to characterize the equilibrium outcomes achievable when players are allowed to communicate prior to playing
a fixed game, and where communication is captured by players' interaction with a mediator.  
\cite{myerson1991game} 
and \cite{forges1993five} provide useful overviews.
Our paper builds on this model by endowing the mediator with the  specific aim of maximizing (some of) the players' utilities.

A different setting, called Bayesian persuasion \citep{kamenica2011bayesian\ifnum\ITCS=1 ,dughmi2019algorithmic\fi}, is one in which only the mediator (here called the sender) has payoff-relevant information, and in which he can commit to a particular information structure prior to observing that information.
Initial work on Bayesian persuasion focused on the case of a single sender and a single player \citep{kamenica2011bayesian}, 
but more recent
research also consists of settings with multiple senders (such as \cite{gentzkow2016competition}) and multiple 
players.
Our paper is most closely related to the latter, and specifically to 
\cite{bergemann2016bayes} and  \cite{galperti2018dual}, who develop a linear programming approach
to study the effect of different information structures on the Bayes correlated equilibria of the subsequent game,
as well as \cite{mathevet2020information}, who extend the belief-based approach of \cite{kamenica2011bayesian}
to study optimal information structures under different solution concepts.

A closely related strand of the literature, recently surveyed by \cite{bergemann2019markets}, focuses on the sale of information 
by a data broker. Most of this work differs from our paper in that it focuses on a market with one-sided information flow, with
information going only from one party to another, whereas we study a market where information flow is bidirectional
across firms. Nonetheless, the insights arising from that literature are related to ones we develop in this paper. 
For example, \cite{bimpikis2019information} consider a monopolistic data broker who sells data to firms who compete
in a downstream market. A major insight in that paper is that the amount of information optimally sold by the data broker
depends on the nature of downstream competition, and, in particular, that more (resp., less) data is sold if that competition
features strategic complements (resp., substitutes). While our paper studies data sharing as opposed to data sale, our results also shed light on the effect the nature of competition has on the amount of data shared. In our model firms' actions are strategic substitutes---when more firms offer the correct good to a particular consumer (e.g., due to the presence of an Amazon), the value of offering the correct good decreases. Our results indicate that sometimes there is more data sharing in the absence of an Amazon and hence when
strategic substitutes are weaker, which is consistent with the insight of \cite{bimpikis2019information}.
However, we also show that sometimes the opposite holds.

There is also a relationship between information sale and information sharing, which is pointed out by \cite{bergemann2019markets}. These authors  relate their general model of a market for information to incentives for information sharing, and point to the study of finely tuned sharing schemes as an open problem.

Our paper is also related to a set of papers that focuses on designing mediators to achieve various
goals, such as to improve the incentives of players, make equilibria robust to collusion, or implement correlated equilibria while guaranteeing privacy  \cite{monderer2004k,monderer2009strong,kearns2014mechanism}.
The first two differ from our work in that they make stronger assumptions about the mediator's capabilities, such as changing payoffs or limiting player actions,
and the third focuses on a setting with many players that is quite different from our own.

Finally, our work is conceptually related to  research on federated learning, and in particular on 
cross-silo federated learning \cite{kairouz2019advances}. This framework consists of a set of agents with individual data
whose goal is to jointly compute a predictive model. A recent emphasis within federated learning is on incentivizing the
agents to participate, sometimes through monetary transfers \cite{zhan2021survey} and sometimes by providing different models to different agents \cite{lyu2020towards}. The focus of these
papers is on fairness---agents that provide more data should be compensated more generously. Our paper differs, in 
that agents are assumed to utilize the resulting models in some competition, and this subsequent competition drives
agents' incentives to participate. Our work is thus orthogonal to that 
surveyed in \cite{zhan2021survey}: we do not focus on the algorithmic aspects of computing a joint model, but rather
on the incentives and joint benefits of participating even when other agents are competitors.
    
\section{Model and Preliminaries}\label{sec:model}
There is a set $G$ of goods and a population of consumers interested in obtaining one of them. Each consumer has one of a finite set of types, $\omega\in\Omega$,
that describe the good $g_\omega\in G$ in which he is interested.
There are three players who compete for consumers: two regular players indexed 1 and 2, and an Amazon, a player indexed 0. We will separate the analysis to two
settings: first, when the Amazon player 0 is not present and players 1 and 2 compete only with one another; and second, when they additionally compete with Amazon.
The model and definitions here apply to both settings.

Player 0 (if present) has complete information of the consumer's type. In Section~\ref{sec:baseline} we assume that
players 1 and 2 have {\em jointly complete information}---that, when combined, their respective data uniquely
identifies each consumer's type---but that each player on her own may only be partially
informed. In Section~\ref{sec:general} we remove this assumption.

To model this informational setting, we represent players' data using the information partition model of \citet{aumann1976agreeing}: 
Each player $i$ is endowed with a partition $\Pi_i$ of $\Omega$, where $\Pi_i$ is a set of 
disjoint, nonempty sets whose union is $\Omega$. 
For each $\omega\in\Omega$ we denote by $P_i(\omega)$ the unique element of $\Pi_i$ that
contains $\omega$, with the interpretation that if the realized type of consumer is $\omega$, each player $i$ only learns that the type belongs to the
set $P_i(\omega)$. 

Framing the example from the introduction within this model would associate $\Omega$ with $\{0,1\}^2$ and the partitions $P_1(00)=P_1(01) = \{00,01\}$, 
$P_1(10)=P_1(11) = \{10,11\}$, $P_2(00)=P_2(10) = \{00,10\}$, and $P_2(01)=P_2(11) = \{01,11\}$.

In this model, player 0's complete information means that $P_0(\omega)=\{\omega\}$ for all $\omega\in\Omega$, and players 1 and 2's jointly complete information
means that $P_1(\omega)\cap P_2(\omega)=\{\omega\}$ for all $\omega\in\Omega$.
We further assume that, before obtaining any information, all players have a common prior $\pi$ over $\Omega$.

To model data sharing between players 1 and 2 we suppose there is a mediator that gathers each player's information
and shares it with the other in some way. Formally, a mediator is a function $\Med:2^\Omega\times 2^\Omega\mapsto \Delta\left(M^2\right)$, 
where $M$ is an arbitrary message space. The range is a distribution over pairs of messages, where the first (respectively, second)  is the message sent to player 1 (respectively, player 2).

We begin with an informal description of the game:
Players offer consumers a good, and consumers choose a player from whom to acquire the good. Consumers are {\em single-minded}: For each $\omega$ there is a unique $g_\omega\in G$ such
that the consumer will only choose a player who offers good $g_\omega$.
 If there is more than one such player, the consumer chooses
uniformly at random between them. One interpretation of this consumer behavior is that he chooses by ``satsificing''---making a random choice among options that
are ``good enough'' \cite{simon1956rational}, where $g_\omega$ represents these good-enough options. See \cite{ben2019regression} for a similar approach.

We assume that prices and costs are fixed, and normalize a player's utility to 1 if she is chosen and to 0 otherwise.
So, for example, in the online advertising application of our model, consumers are assumed to click on a random
ad amongst those ads that are most relevant; once they click, the expected profit to the advertiser (conversion rate times
profit from a sale) is 1. The normalization to 1 is actually without loss of generality. If we had a different profit $p_\omega$ for each consumer type $\omega$, then we could change every $p_\omega$ to 1, modify the prior
probability $\pi(\omega)$ of type $\omega$ to $p_\omega\cdot\pi(\omega)$, and then renormalize the prior to be a proper distribution. This would leave all of our analysis intact.

The following is the order of events, given a fixed mediator $\Med$. It applies both to the case in which there are only two players $\{1,2\}$ and to the case in which there is an additional Amazon player, indexed 0.

\begin{enumerate}
\item Each of players 1 and 2 chooses whether or not to opt into using the mediator.
\item If one or both players $\{1,2\}$ did not opt into using the mediator, then each player $i$ obtains her {\em base value},
the utility $v_i$ (described below).
\item If both players $\{1,2\}$ opted into using the mediator, then:
\begin{enumerate}
\item A consumer of type $\omega$ is chosen from $\Omega$ with prior distribution $\pi$.
\item If present, player 0 learns $P_0(\omega)$.
\item Messages $(M_1,M_2)$ are chosen from the distribution $\Med(P_1(\omega),P_2(\omega))$, and each $i\in\{1,2\}$ learns
 $P_i(\omega)$ and $M_i$.
\item Each player simultaneously chooses a good to offer the consumer, 
and then the consumer chooses a player from whom to obtain the good.
\end{enumerate} 
\end{enumerate}


There are several ways to interpret the base values $v_1$ and $v_2$. Our main interpretation is that these are the expected utilities of the players in the game without data sharing. To formalize this, consider the {\em unmediated} game
$\Gamma= \left(\Omega, \mathcal{P},G^{\abs{\mathcal{P}}},(P_i(\cdot))_{i\in\mathcal{P}},(u_i)_{i\in\mathcal{P}},\pi\right)$,
where $\mathcal{P}$ is the set of participating players, and is either $\{1,2\}$ or $\{0,1,2\}$, $G$ is the set of actions,
$P_i(\cdot)$ is the information of player $i$ (which consists of the partition element of the realized type), $u_i$ is $i$'s utility function (described below), and $\pi$ is the common prior over $\Omega$. 
Given this Bayesian game, let $v=(v_1,v_2)$ be the expected utilities in some {\em Bayesian Nash equilibrium 
(BNE)} of $\Gamma$---a profile $(s_1,s_2)$ in which $s_i$ is the best strategy for each player $i$ conditional on her information and the assumption that the other player plays the strategy $s_{j}$.

Although this is the main interpretation of the base values, our model and results permit additional and more general interpretations as well. 
One additional interpretation is to suppose each $v_i$ is the minimax value of player $i$ in the unmediated game. In this second interpretation, we could
imagine player $j$ ``punishing'' player $i$ if the latter does not opt into using the mediator, by playing a strategy
that minimizes the latter's payoff (off the equilibrium path). A third, more general interpretation is that the base values also depend on factors outside of the specific game, such as firm size, customer base, and so on, in addition to the primitives of the game $\Gamma$.

In our constructions of mediators in Sections~\ref{sec:baseline} and~\ref{sec:general}, the interpretation of $v$ will not matter---our mediators will be optimal given any such values, regardless of whether they are derived endogenously as the equilibrium utilities of the game absent a mediator, or whether they arise exogenously from factors outside of the specific game. However, in Section~\ref{sec:share-more}, in which we compare data sharing across two different environments---with and without a mediator---we will stick with the main interpretation of base values as equilibrium payoffs absent a mediator. This is because, when we compare two different environments, the payoffs players could attain in equilibrium without a mediator will differ across these environments.

Now, as noted above, data sharing is modeled by a mediator $\Med:2^\Omega\times 2^\Omega\mapsto \Delta\left(M^2\right)$. Without loss of generality we invoke the revelation principle and assume that $M= G$,
the set of possible goods.
We interpret the messages of the mediator
as recommended actions to the players, one recommendation for each player. 

Formally, if both players 1 and 2 opt into using the mediator $\Med$, then all
play the mediated Bayesian game
$\Gamma^\Med= \left(\Omega, \mathcal{P},G^{\abs{\mathcal{P}}},\mathcal{I}^\Med,(u_i)_{i\in\mathcal{P}},\pi\right)$.
$\mathcal{P}$ is the set of participating players, and is either $\{1,2\}$ or $\{0,1,2\}$.
The function $\mathcal{I}^\Med:\Omega\mapsto \Delta(T^{\abs{\mathcal{P}}})$ denotes the information of players in the game
for each state, where $T=2^\Omega\times M$ is the set of possible pieces of information a player may have---a partition element and a
message from the mediator.
For player 0, this information consists of the realized type of consumer only,
and so $\mathcal{I}^\Med(\cdot)_0 = \left(P_0(\cdot),\emptyset\right)$ (where the $\emptyset$ means player 0 gets no message from
the mediator). For players 1 and 2, the information consists of both the
partition element of the realized type of consumer and the action recommended to her by the mediator, and so
$\mathcal{I}^\Med(\cdot)_i = \left(P_i(\cdot), \Med(P_1(\cdot),P_{2}(\cdot))_i\right)$.

%
%
%
%

Next, each player's set of actions  in the mediated game is $G$, and her utility function $u_i=u_i:\Omega\times G^{\abs{\mathcal{P}}}\mapsto\R$.
The latter is equal to 0 if player $i$'s action $g\neq g_\omega$, and otherwise it is equal to $1/k$, where $k$ is the total number of players who play action $g_\omega$.
The nonzero utility corresponds to utility 1 if a consumer chooses the player's good, which occurs if a player offers the consumer the correct good and the consumer
chooses uniformly amongst all players that do so. We often write $E[u_i(\cdot)]$ when the expectation is over
the choice of $\omega$, in which case we omit the dependence of $u_i$ on $\omega$ for brevity.
Finally, (mixed) strategies of players in $\Gamma$ are functions $s_i:T\mapsto \Delta(G)$. We also denote by $s_i(\omega)=s_i(\mathcal{I}_i(\omega))$.
 
%
%
 
An important note about the mediator is in order. We assume that when players opt into participating, they {\em truthfully} reveal their data
$P_i(\omega)$ to him. Players' strategic behavior is relevant in their choice of opting in or not, and then in whether or not they follow the
mediator's recommendation in their interaction with the consumer. While truthful reporting is clearly restrictive, it is a natural assumption in our context---for instance, when service providers act as mediators, they are typically also the ones
who host retailers' data in the cloud, and so already have the (true) data on their servers.
Similarly, federated machine learning algorithms work based on the assumption that participants share their true data.
Furthermore,  truthful reporting can be justified by the repeated nature of the interaction. Observe that if a firm
is not truthful in its reporting, the other firm will discover this at some point, as the mediator's recommendations will not pan out
as well as they should. Thus, if firms have the ability
to opt out of using the mediator at some point in the future, then each firm can use the threat of opting out to incentivize
the other to report truthfully.

\paragraph{Individual rationality, incentive compatibility, and optimality}  A first observation is that player 0,
if present, has a dominant strategy:
\begin{lemma}
In any game $\Gamma^\Med$, player 0 has the unique dominant strategy  $s_0(\omega)=g_\omega$.
\end{lemma}
The proof is straightforward: playing action $g_\omega$ leads to positive utility to player 0,
regardless of the other firms' actions. Playing any other action leads to utility zero. This implies that always offering the correct good  
is the unique dominant strategy for player 0. We will thus take that player's strategy as fixed throughout.

We will be interested in designing mediators
that players actually wish to utilize, and so would like them to satisfy two requirements: first, that players want to follow the mediator's recommendations, and so
that following these recommendations forms an equilibrium;
and second, that players prefer their payoffs with the mediator over their payoffs without (the base values). 
Formally,  we say that:
\begin{itemize} \item $\Med$ is {\em incentive compatible (IC)}
if the strategy profile $\overline{s}$ in which players 1 and 2 always follow $\Med$'s recommendation is a BNE of $\Gamma^\Med$.
\item $\Med$ is {\em individually rational (IR)}
if the expected utility of each player $i\in\{1,2\}$ under $\overline{s}$ in $\Gamma^\Med$ is (weakly) greater than $v_i$.
\end{itemize}
Observe that if a mediator is both IC and IR, then in equilibrium players will opt in and then follow the mediator's recommendations.

\Rnote{LEAVE THE FOLLOWING TEXT?
Note that individual rationality is an ex ante notion: 
Players decide whether to opt into the mediator {\em before} they learn $P_i(\omega)$, and if they
do then they are committed. Thus, no player can infer anything about the state from another player's decision about whether or not to opt into the
mediator. }

Finally, as we are interested in the extent to which data sharing  benefits players 1 and 2, we will study {\em optimal} mediators, namely, ones
in which the sum of these two players' utilities is maximal subject to the IC and IR constraints. To this end, denote by $W(\Med)$ the sum of players 1 and 2's 
expected utilities under $\overline{s}$ in $\Gamma^\Med$. Then:
\begin{definition}
Mediator $\Med$ is {\em optimal} if 
$W(\Med)\geq W(\Med')$ for any other 
$\Med'$ that is IC and IR.
\end{definition}

\section{Jointly Complete Information}\label{sec:baseline}
\smallskip
\subsection{Coopetition Without an Amazon}\label{sec:without-amazon}
We begin our analysis with the case of jointly complete information and without an Amazon.
We first consider the question, under what base values $v$ does there exist an IR, IC mediator?
We then turn to our main result for the section, the construction
of an optimal mediator.


\subsubsection{When does an IR, IC mediator exist?} We begin with a simple observation. Since for any mediator $\Med$ and 
conditional on any realized
type $\omega$ the total payoff to players is at most 1, it must be the case that $W(\Med)\leq 1$. Thus,
if $v_1+v_2>1$, then there does not exist any IR mediator.

However, even if $v_1+v_2\leq 1$, there may not exist a mediator that is both IR and IC. In the example from the
introduction, for instance, in the simple case without an Amazon and where $\alpha>2\beta\geq 2\gamma>4\delta$, 
the row retailer can guarantee herself a payoff of $(\alpha+\gamma)/2$ by ignoring the mediator and playing her dominant actions. This implies that
the most the column retailer can obtain is $1-(\alpha+\gamma)/2$. 
Any mediator that leads to a higher payoff to the column
retailer must therefore fail to satisfy the IC constraints.

Theorem~\ref{thm:noA-IR} below generalizes the bound in this example. But first, some notation:
Consider the strategy $s_j'$ of player $j$ that,
for every $P_j(\omega)$, chooses one of the goods that is most likely to be correct: $s_j'(\omega)\in \arg\max_{g\in G}\Pr[g=g_\omega|P_j(\omega)]$. Note that $s'_j$ is a best-response to a player $i\neq j$ who always chooses the
correct good $g_\omega$.
Furthermore, let $\alpha_j$ be the overall probability that $s_j'$ chooses the correct good: 
$\alpha_j=\Pr\left[s'_j(\omega)=g_\omega\right]$.

\begin{theorem}\label{thm:noA-IR}
Suppose $v_i\geq v_j$. Then
an IR, IC mediator exists only if $v_1+v_2\leq 1$ and $v_i\leq E\left[u_i(g_\omega, s_j'(\omega))\right]$.
\end{theorem}
In words, the theorem states that an IR, IC mediator exists only if 
each player's base value is bounded above by her utility in the hypothetical situation in which she always
offers the correct good, and the other player best-responds conditional only on $P_j(\omega)$ (in addition to the
obvious constraint $v_1+v_2\leq 1$). 
In Theorem~\ref{thm:without-amazon} below we will see that these conditions
are also sufficient. 

We note that if the base values $v$ are the equilibrium expected utilities of the game without a mediator, then the conditions in
Theorem~\ref{thm:noA-IR} are always satisfied. To see this, observe that the mediator that recommends to each player her equilibrium action (and nothing more) is IC, by the definition of a BNE. It is also IR, since players' utilities are the same whether or not they opt into using this mediator.

\Rnote{How does a player's information help? By lowering the maximum possible $v_i$ to the other player, and hence
 increasing the bound in $v_j$, namely $1-v_i$.}

\begin{proof}{Proof.}
The condition $v_1+v_2\leq 1$ must hold for any IR mediator $\Med$, since $W(\Med)\leq 1$.
Now, suppose towards a contradiction that there exists an IC mediator for which 
$v_i> E\left[u_i(g_\omega, s_j'(\omega))\right]$. This implies that, under $\Med$, the expected utility
of player $j=3-i$ is strictly less than $1-E\left[u_i(g_\omega, s_j'(\omega))\right]$.  However, player $j$
then has a profitable deviation to $s_j'$. When playing $s_j'$, on each partition element $P_j(\omega)$ 
player $j$ will get utility {\em at least}  
$$\Pr\left[s'_j(\omega)=g_\omega\right]/2 = E\left[u_j(s_j'(\omega), g_\omega)\right]\geq 1-E\left[u_i(g_\omega, s_j'(\omega))\right],$$
and exactly this utility whenever $i$ always chooses the correct good $g_\omega$. 
Since this is a profitable deviation, the mediator $\Med$ is not IC, 
a contradiction.
 \end{proof}
For the rest of Section~\ref{sec:without-amazon}, denote a pair $v$ that satisfies the conditions
in Theorem~\ref{thm:noA-IR} as a {\em feasible $v$}.

\subsubsection{An optimal mediator}
We now describe an optimal mediator, at the same time showing that the conditions on $v$ in Theorem~\ref{thm:noA-IR}
are also sufficient.
We begin with a simple observation:
\begin{observation}\label{obs:noA-opt}
If $\Med$ is IR, IC,  and $W(\Med)=1$, then $\Med$ is optimal.
\end{observation}
This observation implies that the mediator in the example of the introduction, for the case without an Amazon,
is optimal (when the base values are retailers' utilities under their dominant strategies). One drawback of that
mediator, however, is that it depends on the retailers' strategies in the unmediated game---in particular, it utilizes the
fact that, in equilibrium, neither retailer offers the correct good in the bottom-right cell of the matrix. 
Our aim here is to design
a mediator that does not depend on such strategies, but only on players' respective information and base values.

The idea underlying the mediator we construct is the following. Suppose $v_i\geq v_j$. Then if $v_i \leq 1/2$, then the mediator will facilitate
full data-sharing,
and so each player will offer the correct good in every state. If $v_i>1/2$ then player $i$ will obtain all the information and player $j$ will obtain only partial information:
she will get recommendation $g_\omega$  with  probability less than 1, and otherwise will get recommendation  $s_j'(\omega)$. 
We now formally describe the mediator, and then show that it is indeed IR, IC,  and optimal.
%
\begin{algorithm}[H]
\caption{Mediator for $\mathcal{P}=\{1,2\}$ with jointly complete information}\label{alg1}
\begin{algorithmic}[1]
\Procedure{$\Med^v_{noA}$}{$V,W$} \Comment{$V=P_1(\omega)$ and $W=P_2(\omega)$ for realized $\omega\in\Omega$}
\State $\omega \leftarrow V\cap W$
\State $i \leftarrow \arg\max_{k\in\{1,2\}} v_k$ 
\State $j \leftarrow 3-i$
   \If{$v_i\leq 1/2$}
	\Return $\left(g_\omega, g_\omega\right)$	 \Comment{Full data-sharing}
  \Else
	\State $g_j \leftarrow s_j'(\omega)$
	\State Choose $\gamma \in \left[0,1\right]$ uniformly at random.
	\If{$\gamma <\frac{2-2v_i-\alpha_j}{1-\alpha_j}$}
		\State \Return $(g_\omega,g_\omega)$
	\ElsIf{$i=1$} \Return $(g_\omega, g_j)$
		\Else~\Return $(g_j, g_\omega)$
	\EndIf
\EndIf
\EndProcedure
\end{algorithmic}
\end{algorithm}

\begin{theorem}\label{thm:without-amazon}
For any feasible $v$, the mediator $\Med^v_{noA}$ above is IR, IC,  and optimal. 
\end{theorem}

\begin{proof}
First, consider the case in which $v_i\leq 1/2$. In this case, strategy $\overline{s}$ in $\Gamma^{\Med^v_{noA}}$ leads to expected utility $1/2$ to each player, and so the mediator is IR.
Furthermore, the strategy $\overline{s}$ is a BNE, since each player always
chooses $g_\omega$, which is dominant, and so the mediator is also IC.
Finally, since in this case  $W\left(\Gamma^{\Med^v_{noA}}\right)=1$, Observation~\ref{obs:noA-opt} implies
that the mediator is optimal.

Next, suppose $v_i>1/2$. Observe that, conditional on line 10 being activated, the expected utility of player $j$
is $1/2$, whereas conditional on line 11 or 12, the expected utility of player $j$ is $\alpha_j/2$ (since with probability
$\alpha_j$ she chooses the correct good, and player $i$ always chooses the correct good). Thus,
overall the
 expected utility of player $j$ is 
 $$\frac{2-2v_i-\alpha_j}{1-\alpha_j}\cdot\frac{1}{2}  
 +\left(1-\frac{2-2v_i-\alpha_j}{1-\alpha_j}\right)\cdot\frac{\alpha_j}{2} = 1-v_i.$$
Since player $i$ always chooses the correct good, it must be the case that $W\left(\Gamma^{\Med^v_{noA}}\right)=1$,
and so the utility of player $i$ is $v_i$. Finally, since $v_1+v_2\leq 1$, this implies that the utility of player $j$ 
is at least $v_j$. Thus, the mediator is IR. 
 
Next, observe that $\overline{s}$ in $\Gamma^{\Med^v_{noA}}$ is dominant for player $i$, since she always chooses $g_\omega$.
It is also optimal for player $j$, since she either also chooses $g_\omega$, or chooses the best response conditional on her information $P_j(\omega)$. Thus, the strategy $\overline{s}$ is a BNE, and so the mediator is IC.

Finally, since $W\left(\Gamma^{\Med^v_{noA}}\right)=1$, Observation~\ref{obs:noA-opt} implies
that the mediator is optimal.
 \end{proof}

\subsection{Coopetition Against an Amazon}\label{sec:with-amazon}
In this section we analyze the game with an Amazon.
As in Section~\ref{sec:without-amazon}, we first consider the question, under what base values $v$ does there exist an IR, IC mediator?
We then turn to our main result for the section, the construction of an optimal mediator. 

\subsubsection{When does an IR, IC mediator exist?} Once again, we begin with a simple observation. 
Since for any mediator $\Med$ and any
type $\omega$ the total welfare of players 1 and 2 is at most $2/3$ (since Amazon always knows the type,
offers the correct good, and gets at least $1/3$ of the surplus), it must be the case that $W(\Med)\leq 2/3$. Thus,
if $v_1+v_2>2/3$, then there does not exist any IR mediator.

However, and again as above, even if $v_1+v_2\leq 2/3$ there may still not exist a mediator that is both IR and IC. In the example from the
introduction, for instance, in the simple case with an Amazon and where $\alpha>2\beta\geq 2\gamma>4\delta$, 
the row retailer can guarantee herself a payoff of $(\alpha+\gamma)/3$ by ignoring the mediator and playing her dominant actions. This implies that
the most the column retailer can obtain is $2/3-(\alpha+\gamma)/3$. 
Any mediator that leads to a higher payoff to the column
retailer must therefore fail to satisfy the IC constraints.

The result below consists of bounds on the base values under which an IR, IC mediator exists. 

\begin{theorem}\label{thm:A-IR}
Suppose $v_i\geq v_j$. Then
an IR, IC mediator exists only if 
$v_i\leq E\left[u_i(g_\omega, s_j'(\omega))\right]$ and $v_j\leq 1-2v_i$.
\end{theorem}

As with Theorem~\ref{thm:noA-IR}, we note that if the base values $v$ are the equilibrium expected utilities of the game without a mediator, then the conditions in
Theorem~\ref{thm:A-IR} are satisfied. Again, this follows from the observation that the mediator that recommends to each player her equilibrium action is both IR and IC.


\Rnote{Maybe note again, how does a player's information help? By lowering the maximum possible $v_i$ to the other player, and hence increasing the bound in $v_j$, namely $1-2v_i$.}

Before proving the theorem we develop some notation and intuition.
For any mediator $\Med$ and player $k\in\{1,2\}$ let
$\beta_k^\Med = \Pr\left[(\Med(\omega)_k=g_\omega)\cap (\Med(\omega)_{3-k}\neq g_\omega)\right]$ and 
$\beta^\Med=\Pr\left[\Med(\omega)_k=\Med(\omega)_{3-k}= g_\omega\right]$.
With this notation, the expected utility of player $k$ under $\overline{s}$ in $\Gamma^\Med$ is $\beta^\Med_k/2 + \beta^\Med/3$: For any $\omega$, if player $k$ offers $g_\omega$
and the other player, $3-k$, does not, then the consumer chooses between players $k$ and 0, and so each gets $1/2$. If both offer $g_\omega$, 
the consumer chooses between players $k$, $3-k$, and 0, and so each gets $1/3$. 

An IR, IC mediator $\Med$ is one for which $\overline{s}$ is a BNE, and in which $\beta_k^\Med/2 + \beta^\Med/3 \geq v_k$ for both players.
In addition to the incentive constraints, we also have the constraints $\beta_i^\Med+\beta_j^\Med +\beta^\Med\leq 1$, and all three variables non-negative. 

Finally, the total welfare of players 1 and 2 in $\Med$ is
$W(\Med)=(\beta^\Med_1+\beta^\Med_2)/2 + 2\beta^\Med/3$, and so an optimal mediator is one that maximizes this sum subject to the IC and IR constraints above. Denote the linear program above, disregarding the IC constraint, 
as $\lp(v_1,v_2)$:

\begin{equation*}
\begin{array}{lr@{}ll}
\text{maximize}  & \frac{\beta_1+\beta_2}{2} + \frac{2\beta}{3} &\\
\text{subject to}& \frac{\beta_i}{2} + \frac{\beta}{3} &\geq v_i,  &i=1,2\\
                 &                                                \beta_i&\geq 0, &i=1 ,2\\
                                  &                                                \beta&\geq 0 &\\
                                  & \beta_1+\beta_2+\beta& \leq 1&

\end{array}
\end{equation*}

To prove Theorem~\ref{thm:A-IR}, we need to identify conditions on $v$ under which the LP is feasible. 
To do this, it will be helpful to begin with some lemmas that characterize the LP's optimal solution:

\begin{lemma}\label{lem:LP}
Suppose $v_i\geq v_j$ and that  $\lp(v_1,v_2)$ is feasible. Then any optimal solution $(\beta_1,\beta_2,\beta)$ to 
$\lp(v_1,v_2)$ satisfies:
\begin{itemize}
\item[(a)] $\beta_1+\beta_2+\beta=1$.
\item[(b)] $\beta_j=0$.
\item[(c)] If $v_i>1/3$ then $i$'s IR constraint binds: $\beta_i/2 + \beta/3 = v_i$.
\end{itemize}
\end{lemma}

\begin{proof}
Suppose (a) does not hold, and that $\beta_1+\beta_2+\beta<1$.  This means we can increase $\beta$ to $\beta'>\beta$ without violating any of the constraints.
The value of the $\lp$ is now $(\beta_1+\beta_2)/2 + 2\beta'/3>(\beta_1+\beta_2)/2 + 2\beta/3$, contradicting the optimality $(\beta_1,\beta_2,\beta)$.

Now suppose (a) holds, but (b) does not. If $\beta_i<\beta_j$ then the solution $(\beta_2,\beta_1,\beta)$ is also feasible
and optimal, so assume without loss of generality that $\beta_i\geq \beta_j$.
Consider now the solution $\beta_i' = \beta_i-\beta_j$, $\beta_j'=0$, and $\beta'=\beta+2\beta_j$. Since
$$\beta_i'/2+\beta'/3 = (\beta_i-\beta_j)/2 + (\beta+2\beta_j)/3 > \beta_i/2+\beta/3 \geq v_i$$ and 
$$\beta_j'/2+\beta'/3 = (\beta+2\beta_j)/3 > \beta_j/2+\beta/3 \geq v_j,$$
this new solution is feasible. Moreover,
$$(\beta_1'+\beta_2')/2 + 2\beta'/3>(\beta_1+\beta_2)/2 + 2\beta/3,$$
contradicting the optimality of  $(\beta_1,\beta_2,\beta)$.

Now suppose (a) and (b) hold, but that (c) does not. Then there is some $\eps>0$ such that $\beta_i'/2 + \beta'/3 > v_i$,
where $\beta_i'=\beta_i-\eps$ and $\beta'=\beta+\eps$, and so $(\beta_1',\beta_2',\beta')$ is also a feasible solution. Furthermore,
$$\beta_i'/2 + 2\beta'/3>\beta_i/2 + 2\beta/3,$$
contradicting the optimality of  $(\beta_1,\beta_2,\beta)$.
 \end{proof}

One implication of Lemma~\ref{lem:LP} is that at an optimal solution, $\beta$ is maximized. This follows from the observation that
the objective function is
$(\beta_1+\beta_2+\beta)/2 + \beta/6 = 1/2+\beta/6$ when $\beta_1+\beta_2+\beta =1$.
In addition, Lemma~\ref{lem:LP} pins down the value of the optimal solution to $\lp(v_1,v_2)$:
\begin{lemma}\label{lem:opt-LP}
Suppose $v_i\geq v_j$ and that  $\lp(v_1,v_2)$ is feasible. Then the value $\opt(v_1,v_2)$ of an optimal solution to $\lp(v_1,v_2)$ is one of the following:
\begin{itemize}
\item If $v_i\leq 1/3$ then $\opt(v_1,v_2)=2/3$.
\item If $v_i>1/3$ then $\opt(v_1,v_2) = 1-v_i$.
\end{itemize}
\end{lemma}

\begin{proof}
Lemma~\ref{lem:LP} implies that at the optimal solution we have $\beta_i+\beta=1$, and so the value at the optimum is 
$\beta_i/2 +2\beta/3$.
In the first case, the optimal solution occurs at the point $\beta_1=\beta_2=0$, $\beta=1$: it maximizes the objective, while both
IR  constraints remain slack since $v_i\leq 1/3 = \beta/3$.

In the second case, Lemma~\ref{lem:LP} implies that  $\beta_i/2 + \beta/3 = v_i$ and that $\beta_i=1-\beta$.
These imply that $v_i=(1-\beta)/2 + \beta/3$, and so that $\beta=3-6v_i$ and $\beta_i = 1-3+6v_i$.
Plugging these into the objective yields
$$\opt(v_1,v_2) = \frac{\beta_i}{2}+\frac{2\beta}{3} = \frac{1-3+6v_i}{2}+\frac{2(3-6v_i)}{3} = 1-v_i.$$
 \end{proof}

Lemmas~\ref{lem:LP} and~\ref{lem:opt-LP} characterize the maximal values attainable by an IR mediator, and are the analogues
of the condition that $v_1+v_2\leq 1$ when there is {\em no} Amazon. However, before proving Theorem~\ref{thm:A-IR}, we need two additional lemmas.
Fix a mediator $\Med$, and consider
the hypothetical situation in which  player $i$ obtains additional information,
and in each state
$\omega$ she fully learns that the state is $\omega$ and so plays the dominant strategy $\hat{s}_i(\omega)=g_\omega$. In this new game, denote the best response of player $j$ as $s'_j(\Med)$, the strategy that, for every $P_j(\omega)$ and $\Med(\omega)_i$, chooses some $g$ that is most-likely correct conditional on $P_j(\omega)$ and $\Med(\omega)_i$.
Also observe that here, player $j$ is never the only player to choose $g_\omega$, since player $i$ always plays $g_\omega$. 

Player $j$ is worse off in this hypothetical game than when all follow $\Med$'s
recommendations:
\begin{lemma}\label{lem:util_j}
$E[u_j(s'_j(\Med), \hat{s}_i)] \leq E[u_j(\overline{s}_j, \overline{s}_i)]$ for any mediator $\Med$.
\end{lemma}

\begin{proof}
Fix a particular state $\omega$. For every action $g$ of player $j$, who conditions on $P_j(\omega)$ and
$\Med(\omega)_j$, the utility 
$E[u_j(g, \hat{s}_i)|P_j(\omega),\Med(\omega)_j] \leq E[u_j(g, \overline{s}_i)|P_j(\omega),\Med(\omega)_j]$, since $i$ makes correct choices more often in the former. This implies that
$E[u_j(s'_j(\Med), \hat{s}_i)|P_j(\omega),\Med(\omega)_j] \leq E[u_j(\overline{s}_j, \overline{s}_i)|P_j(\omega),\Med(\omega)_j]$, since in both sides of the inequality $j$ chooses an optimal action
conditional on her information, but in the RHS the utility from every choice $g$ is higher than its corresponding utility in the LHS. 
This implies that $E[u_j(s'_j(\Med), \hat{s}_i)] \leq E[u_j(\overline{s}_j, \overline{s}_i)]$.
 \end{proof}

Although in the hypothetical game player $j$ is worse off, the sum of players' utilities is higher:
\begin{lemma}\label{lem:sum-utils}
The sum of players' utilities is higher under $(s'_j(\Med), \hat{s}_i)$ than under $(\overline{s}_j, \overline{s}_i)$ in $\Med$. 
\end{lemma}

\begin{proof}
Again, we show that this holds conditional on every
$P_j(\omega)$ and $\Med(\omega)_j$. Fix $\omega$, and 
observe that 
$$\Pr\left[\hat{s}_i(\omega)=\overline{s}_{j}(\omega)= g_\omega|P_j(\omega),\Med(\omega)_j\right]\geq \Pr\left[\overline{s}_i(\omega)=\overline{s}_{j}(\omega)= g_\omega|P_j(\omega),\Med(\omega)_j\right],$$
since $\hat{s}_i$ always correctly chooses $g_\omega$. Furthermore, 
$$\Pr\left[\hat{s}_i(\omega)=s_{j}'(\Med)= g_\omega|P_j(\omega),\Med(\omega)_j\right]\geq \Pr\left[\hat{s}_i(\omega)=\overline{s}_{j}(\omega)= g_\omega|P_j(\omega),\Med(\omega)_j\right],$$
since $s_j'$ 
chooses the good that maximizes the probability of choosing $g_\omega$. Thus,
$$\Pr\left[\hat{s}_i(\omega)=s_{j}'(\Med)= g_\omega|P_j(\omega),\Med(\omega)_j\right]\geq \Pr\left[\overline{s}_i(\omega)=\overline{s}_{j}(\omega)= g_\omega|P_j(\omega),\Med(\omega)_j\right].$$

Now recall that $\lp(v_1,v_2)$ is maximized when $\beta$, the probability that both players correctly choose $g_\omega$ is maximized. 
Thus, the value of $\lp(v_1,v_2)$, the sum of players' utilities, is higher under 
$(s'_j, \hat{s}_i)$ than under $(\overline{s}_1,\overline{s}_2)$.
 \end{proof}

We can now prove Theorem~\ref{thm:A-IR}.
\begin{proof} 
%
For the first condition, the upper bound on $v_i$, fix some IR, IC mediator $\Med$. We will argue that
the utility to player $i$ under $\Med$ is at most $E\left[u_i(g_\omega, s_j'(\Med))\right]$. Lemma~\ref{lem:util_j}
implies that player $j$'s utility decreases if player $i$ additionally obtains all the information (even when $j$ responds optimally), and Lemma~\ref{lem:sum-utils} implies that in this same situation the sum of utilities increases. Together,
these lemmas imply that $i$'s utility is higher when she obtains all the information and $j$ only learns $P_j(\omega)$ and $\Med(\omega)_j$, than when $i$ also only learns $P_i(\omega)$ and $\Med(\omega)_i$. Furthermore, when $i$
has all the information,  she is best off when $j$ has the least information---that is, when $\Med(\omega)_j=P_j(\omega)$.
This is because any additional information that causes $j$ to improve upon her payoff will do so at the expense of 
player $i$'s payoff (by transferring weight from $\beta_i^\Med$ to $\beta^\Med$). Thus, in any IR, IC mediator, the highest possible utility to player $i$ is $E\left[u_i(g_\omega, s_j'(\omega))\right]$, since $\Med(\omega)_j=P_j(\omega)$
and so $s_j'(\omega)=s_j'(\Med)$.

The second condition, $v_j\leq 1-2v_i$, follows from Lemma~\ref{lem:opt-LP}. If $v_i\leq 1/3$, then the maximal
welfare of any IR mediator is $2/3$, and so $v_j\leq 2/3-v_i\leq 1-2v_i$. If $v_i>1/3$, then the lemma implies that the
maximal welfare of any IR mediator is at most $1-v_i$. Since player $i$ gets at least $v_i$, player $j$ can get at most $1-2v_i$.
\end{proof}

For the rest of Section~\ref{sec:with-amazon}, denote a pair $v$ that satisfies the conditions
in Theorem~\ref{thm:A-IR} as a {\em feasible $v$}.

\subsubsection{An optimal mediator}

We now describe an optimal mediator, at the same time showing that the conditions in Theorem~\ref{thm:A-IR}
are not only necessary but also sufficient.

We begin with a definition.
\begin{definition}
A mediator $\Med$ is {\em fully revealing to player $i$} if the mediator's recommendation to player $i$ always coincides with the optimal good: 
$\overline{s}_i(\omega)=g_\omega$ for every $\omega$. A mediator facilitates {\em full data-sharing} if it is fully revealing to both players $\mathcal{P}\setminus\{0\}$.
\end{definition}

Fix the set of players to be $\mathcal{P}=\{0,1,2\}$.
The idea underlying the mediator is the same as that of $\Med_{noA}^v$. Suppose $v_i\geq v_j$. 
Then if $v_i \leq 1/3$, then the mediator will facilitate
full data-sharing,
and so each player will offer the correct good in every state. If $v_i>1/3$ then player $i$ will obtain all the information and player $j$ will obtain only partial information:
she will get recommendation $g_\omega$  with  probability less than 1, and otherwise will get recommendation  $s_j'(\omega)$. 
The challenge is to maximize the probability of recommendation $g_\omega$ subject to the IR constraints, and to do this in such a way that following the 
recommendations is a BNE.

We now formally describe the mediator, and then show that it is indeed IR, IC,  and optimal.
The main differences from $\Med_{noA}^v$ in Mediator~\ref{alg1} are lines 5 and 9.
\begin{algorithm}[H]
\caption{Mediator for $\mathcal{P}=\{0,1,2\}$ with jointly complete information}\label{alg2}
\begin{algorithmic}[1]
\Procedure{$\Med^v_A$}{$V,W$} \Comment{$V=P_1(\omega)$ and $W=P_2(\omega)$ for realized $\omega\in\Omega$}
\State $\omega \leftarrow V\cap W$
\State $i \leftarrow \arg\max_{k\in\{1,2\}} v_k$ 
\State $j \leftarrow 3-i$
	\If{$v_i\leq 1/3$}
	\Return $\left(g_\omega, g_\omega\right)$	 \Comment{Full data-sharing}
\Else
	\State $g_j \leftarrow s_j'(\omega)$
	\State Choose $\gamma \in \left[0,1\right]$ uniformly at random.
	\If{$\gamma <\frac{3-6v_i-\alpha_j}{1-\alpha_j}$}
		\State \Return $(g_\omega,g_\omega)$
	\ElsIf{$i=1$} \Return $(g_\omega, g_j)$
		\Else~\Return $(g_j, g_\omega)$
	\EndIf
\EndIf
\EndProcedure
\end{algorithmic}
\end{algorithm}

\begin{theorem}\label{thm:with-amazon}
For any feasible $v$, the mediator $\Med^v_{A}$ above is IR, IC,  and optimal. 
\end{theorem}

An interesting feature of $\Med^v_{A}$ is that it either facilitates full data-sharing, or is fully revealing to one of the players.
These features are inherent to {\em any} optimal IR mediator:
\begin{theorem}\label{thm:opt-mediator}
For any feasible $v$ and any optimal IR mediator $\Med$, one of the following holds:
\begin{itemize}
\item $\Med$ is fully revealing to both players 1 and 2.
\item $\Med$ is fully revealing to one of the players, and the IR constraint binds for that player (i.e., she is indifferent between $\Med$ and her base value).
\end{itemize}
\end{theorem}
Note that Mediator~\ref{alg1} for the case of no Amazon also satisfies
the features described in Theorem~\ref{thm:opt-mediator}. In that case, however, these features do not characterize all mediators,
as there exist mediators that are optimal for the case of no Amazon that do not have them. One such
mediator is the one described in the example from the introduction. With an Amazon, in contrast, these features characterize
all optimal IR mediators.

We begin by showing that the mediator above is optimal.
Our approach is to show that $W(\Med^v_A)$ is equal to the value of the optimal solution to $\lp(v_1,v_2)$, and that the strategy $\overline{s}$ in $\Med^v_A$ is a BNE.
Observe that $\lp(v_1,v_2)$ does not contain the IC constraints, implying that they are not binding, and so only the
IR constraints are relevant for optimality. This
will no longer be the case without the assumption of jointly complete information, analyzed in Section~\ref{sec:general} below.

We already laid most of the groundwork for the proofs of Theorems~\ref{thm:with-amazon} and~\ref{thm:opt-mediator} when we proved Theorem~\ref{thm:A-IR}, but
do need one last lemma:

\begin{lemma}\label{lem:gamma}
Suppose $v$ is feasible and $v_i>1/3$. Then $\alpha_j=\Pr\left[s'_j(\omega)=g_\omega\right]\leq 3-6v_i$.
\end{lemma}

\begin{proof}
Since $v$ is feasible, it holds that $E\left[u_i(g_\omega, s_j'(\omega))\right]\geq v_i$.
Next, note that $E[u_i(g_\omega, s_j'(\omega))]=\alpha_j/3 + (1-\alpha_j)/2$. Together, these facts imply that $\alpha_j\leq 3-6v_i$, as claimed.
%
%
%
 \end{proof}

We are now ready to prove Theorem~\ref{thm:with-amazon}.
\begin{proof} 
First, consider the case in which $v_i\leq 1/3$. In this case, strategy $\overline{s}$ in $\Gamma^{\Med^v_{A}}$ leads to expected utility $1/3$ to each player, which
satisfies the IR constraint. Furthermore,  strategy $\overline{s}$ is a BNE, since each player always
chooses $g_\omega$, which is dominant. Finally, $W\left(\Med^v_{A}\right)$ is maximal by Lemma~\ref{lem:opt-LP}. 

Next, suppose $v_i>1/3$. Observe that $\overline{s}$ in $\Gamma^{\Med^v_{A}}$ is dominant for player $i$, since she always chooses $g_\omega$.
It is also optimal for player $j$, since she either also chooses $g_\omega$, or chooses the best response conditional on her information $P_j(\omega)$. Thus, the mediator is IC.

We now show that $\Med^v_{A}$ is optimal. For simplicity, henceforth denote by $\Med = \Med^v_{A}$. 
First note that $\beta^\Med_j=0$. Next, by Lemma~\ref{lem:gamma}, it holds that $\alpha_j\leq 3-6v_i$. Since
$v_i\in (1/3, 1/2]$, this implies that $\alpha_j\in [0,1)$.
These, in turn imply that  
$$\frac{3-6v_i-\alpha_j}{1-\alpha_j} \in [0,1],$$
so line 9 of Mediator~\ref{alg2} is valid. From that line we have that  
$$\beta^\Med = \alpha_j + (1-\alpha_j)\cdot\frac{3-6v_i-\alpha_j}{1-\alpha_j} = 3-6v_i,$$
since both players choose $g_\omega$ when $s'_j(\omega)=g_\omega$ and when $s'_j(\omega)\neq g_\omega$ but
$\gamma<(3-6v_i-\alpha_j)/(1-\alpha_j)$. This means that 
$$W(\Med)=\frac{\beta_i^\Med}{2}+\frac{2\beta^\Med}{3} =\frac{1-3+6v_i}{2}+\frac{2(3-6v_i)}{3} = 1-v_i,$$
which, by Lemma~\ref{lem:opt-LP}, is the value of the optimal solution to $\lp(v_1,v_2)$.

Finally, since $i$'s utility is $v_i$ and so  $j$'s utility is $1-2v_i$, the mediator is IR.
\end{proof}

Finally, we can also prove Theorem~\ref{thm:opt-mediator}.
\begin{proof} 
Suppose first that $v$ is such that $v_j\leq v_i\leq 1/3$. In this case, full data-sharing is IC, and so any optimal mediator $\Med$
must have $W(\Med)\geq 2/3$. However, the maximum value of $\lp(v_1,v_2)$ is $2/3$, and the only way to obtain this is to have $\beta^\Med=1$.
This is equivalent to full data-sharing, and so the only mediator that is optimal is the full data-sharing one.

Now suppose that $v_i>1/3$. 
By the proof of Theorem~\ref{thm:with-amazon}, an optimal solution to $\lp(v_1,v_2)$ is attainable by some mediator, in particular by $\Med_A^v$. Thus, any mediator
$\Med$ that is optimal must also yield $W(\Med)$ that is equal to the optimal value of $\lp(v_1,v_2)$. By Lemma~\ref{lem:LP}, any optimal solution to $\lp(v_1,v_2)$ must have $\beta_j=0$ and $\beta_i/2 + \beta/3 = v_i$. Thus, the mediator $\Med$ must satisfy these as well: the first equality implying that $\Med$
is fully revealing to player $i$, and the second that player $i$'s IR constraint binds.
\end{proof}

\section{No Jointly Complete Information}\label{sec:general}
We now drop the assumption that players 1 and 2 have jointly complete information. Formally,
for a given type $\omega$ let $S(\omega)=P_1(\omega)\cap P_2(\omega)$, and call each such $S(\omega)$
a {\em segment}. The assumption of jointly complete information states that all segments are singletons,
and in this section we consider the more general setting in which segments may contain more than one type
of consumer.

This setting presents new challenges,
as there is now a new conflict between optimal strategies and welfare. To see this, consider some segment
$S=\{\omega^1,\omega^2\}$ for which $\Pr[\omega^1] =4\cdot \Pr[\omega^2]$. Then, conditional on segment $S$, both players
have a dominant strategy, namely, to offer $g=g_{\omega^1}$. To see this, observe that a player's utility from offering 
$g_{\omega^1}$ is at least $\Pr[\omega^1]/2$ regardless of the other's action, the utility from offering $g_{\omega^2}$ is at most 
$\Pr[\omega^2]$, and by assumption $\Pr[\omega^1]/2 = 2 \Pr[\omega^2]>\Pr[\omega^2]$.
This leads to total welfare 
$\Pr[\omega^1]$. However,
if players were to separate rather than pool---one offering $g_{\omega^1}$ and the other $g_{\omega^2}$---then the total
welfare would be $\Pr[\omega^1]+\Pr[\omega^2]$, which is higher. Note that this conflict between optimal strategies
and welfare does not occur when players have jointly complete information, since then the pooling action, which
is dominant, is also welfare-maximizing. Finally, although we illustrated this conflict for the setting without an
Amazon, it is of course also present with an Amazon.

This conflict between optimal strategies and welfare maximization makes the design of mediators significantly more
challenging. In principle, one could formulate the problem as a complex linear program, whose solution could
be used to find the optimal mediator. Here, however, we take a different approach. We will show that {\em simple}
mediators---in particular, variants of the ones from Section~\ref{sec:baseline} for the setting with
 jointly complete information---are
{\em approximately} optimal. And while this approach entails some loss in terms of optimality, we will show that
these simple mediators are not only approximately optimal relative to the optimal IR, IC mediator, but
rather that they are approximately optimal relative to a higher benchmark---namely, the welfare that can be achieved by {\em any} mediator, even one that
is not IR or IC. Then, to complement our results, we will show that the approximation factors our mediators achieve are tight or nearly-tight,
and that no other IR and IC mediator can in general achieve a better approximation factor relative to that same benchmark.

We begin with some notation. For any given segment $S(\omega)$, denote by 
$g_\omega^1$ the good $g$ that has the highest probability $\pr{g=g_\omega|S(\omega)}$ in $S(\omega)$, 
and by $g_\omega^2$ the good
with the second highest probability, under the prior $\pi$. If there are several such goods, fix two such goods arbitrarily. Next,
let  $\phi_\omega^1=\pr{g_\omega^1=g_\omega|S(\omega)}$ be the probability of $g_\omega^1$ conditional on $S(\omega)$,
and $\phi_\omega^2=\pr{g_\omega^2=g_\omega|S(\omega)}$ be the respective probability of $g_\omega^2$.
We will also use the notation that, for a given segment $S\in\{S(\omega):\omega\in\Omega\}$, the goods $g_S^1$ and $g_S^2$ are the
goods with the highest and second-highest probabilities in $S$, with $\phi_S^1$ and $\phi_S^2$ 
their respective conditional probabilities.
Finally, let $\phi^1$ be the total (unconditional) weight of goods $g_S^1$ over all segments $S$, and $\phi^2$ be
the total weight of goods $g_S^2$: Formally,
$$\phi^k = \sum_{S\in\{S(\omega):\omega\in\Omega\}} \phi_S^k\cdot\pr{S}~~~~~~\mbox{ for each } k\in\{1,2\}.$$
\Rnote{Maybe need only $g_S^k$ and not $g_\omega^k$.}

Now, as our main interest lies in coopetition against an Amazon, we will focus here on the setting with an Amazon.
\Rnote{Maybe: A similar analysis for the setting without an Amazon appears in Appendix~\ref{sec:general-noA}.}
We note that,
while we drop the assumption that players 1 and 2 have jointly complete information, we still assume that
the Amazon player 0 has complete information. That is, even though in any given segment $S$ players
1 and 2 do not know which type $\omega\in S$ is realized, player 0 does know, and offers each such type
the correct good $g_\omega$.

We proceed as follows. We begin in Section~\ref{sec:noJCI-1} with the design of an optimal mediator for a particular 
setting of parameters,
namely, for the case in which $\phi_S^1\leq \frac{3}{2} \cdot \phi_S^2$ for every segment $S$. We then turn to our main analysis
in Section~\ref{sec:noJCI-2},
for the case in which $\phi_S^1> \frac{3}{2}\cdot \phi_S^2$ for every segment $S$. There we design the approximately optimal
simple mediator mentioned above, and analyze the tightness of the approximation factor. 
Finally, in Section~\ref{sec:noJCI-3} we show how to combine the two mediators in order to obtain an approximately optimal 
mediator for the general setting, without restrictions on the segments.

\subsection{Optimal Mediator for $\phi_S^1\leq \frac{3}{2} \cdot \phi_S^2$}\label{sec:noJCI-1}
We begin with the simpler case in which  $\phi_S^1\leq \frac{3}{2} \cdot \phi_S^2$ for every segment $S$.
This case is simpler because here, conditional on any given segment, 
there is no conflict between optimal strategies and welfare maximization. In particular, both involve separating,
with one player offering $g_S^1$ and the other offering $g_S^2$:
\begin{claim}\label{claim:no-conflict}
Conditional on any segment $S$, if $\phi_S^1\leq \frac{3}{2} \cdot \phi_S^2$ then the separating
strategy profile is both welfare maximizing and an equilibrium.
\end{claim}
\begin{proof}
We will compare the only two reasonable player choices for welfare maximization
conditional on segment $S$: (i) pooling on $g_S^1$, or
(ii) separating, with one player offering $g_S^1$ and the other $g_S^2$.
\begin{itemize}
\item[(i)] Pooling: Here each of players 1 and 2 obtains utility $\phi_S^1/3$---when type $\omega\in S$
with $g_\omega = g_S^1$ is realized (which happens with probability $\phi_S^1$), the players split the surplus of 1 with player 0, and so each gets $1/3$. Welfare of players 1 and 2 is thus $\frac{2}{3}\cdot\phi_S^1$.
\item[(ii)] Separating: Here the player who offers $g_S^1$ obtains utility $\phi_S^1/2$---when type $\omega\in S$
with $g_\omega = g_S^1$ is realized, that player splits the surplus with player 0. Similarly, the player  
who offers $g_S^2$ obtains utility $\phi_S^2/2$. Welfare of players 1 and 2 is thus $(\phi_S^1+\phi_S^2)/2$.
\end{itemize}
Observe that, when $\phi_S^1\leq \frac{3}{2} \cdot \phi_S^2$, we have that $(\phi_S^1+\phi_S^2)/2 > \frac{2}{3}\cdot\phi_S^1$, and so the welfare maximizing choices are the separating ones.

In addition to being welfare-maximizing, the separating actions also form an equilibrium (conditional on $S$).
Fix a segment $S$, and suppose player 1 offers $g_S^1$ and player 2 offers $g_S^2$. Then player 1
clearly does not wish to deviate, as she currently obtains utility $\phi_S^1/2$, a deviation to $g_S^2$
will lead to utility $\phi_S^2/3$, and a deviation to any other good will lead to utility at most $\phi_S^1/2$ (since
$g_S^1$ has maximal probability in $S$). Player 2 also does not wish to deviate, as she currently obtains
utility $\phi^2_S/2$, a deviation to $g^1_S$ will lead to utility $\phi_S^1/3$, and by assumption
$\phi^2_S/2\geq \phi_S^1/3$.
 \end{proof}

Given the lack of conflict between welfare maximization and equilibrium shown in Claim~\ref{claim:no-conflict},
we can now construct optimal mediators for this parameter setting. The mediator is straightforward: it
facilitates full data-sharing between the players, and recommends the actions that lead to a separating
equilibrium.
\begin{algorithm}[H]
\caption{Mediator for $\mathcal{P}=\{0,1,2\}$ without jointly complete information}\label{alg3}
\begin{algorithmic}[1]
\Procedure{$\Med^v_{1}$}{$S$} \Comment{$S=P_1(\omega)\cap P_2(\omega)$ for realized $\omega\in\Omega$}
\If{$\phi^1=\phi^2$} \Return $(g_S^1,g_S^2)$
\Else
\State $i \leftarrow \arg\max_{k\in\{1,2\}} v_k$ 
\State $j \leftarrow 3-i$
\State Choose $\gamma \in \left[0,1\right]$ uniformly at random.
\If{$\gamma <\frac{2v_i - \phi^2}{\phi^1-\phi^2}$}
	\If{$i=1$} \Return $(g_S^1,g_S^2)$ 
	\Else~\Return $(g_S^2,g_S^1)$
	\EndIf
\ElsIf{$i=1$} \Return $(g_S^2,g_S^1)$ 
\Else~\Return $(g_S^1,g_S^2)$
\EndIf
\EndIf

%
\EndProcedure
\end{algorithmic}
\end{algorithm}

\begin{theorem}\label{thm:no-conflict}
Suppose $v_1+v_2\leq (\phi^1+\phi^2)/2$ and $\max\{v_1,v_2\}\leq \phi^1/2$. 
Then the mediator $\Med^v_{1}$ above is IR, IC,  and optimal. 
\end{theorem}
Observe that the conditions on $v$ in Theorem~\ref{thm:no-conflict} are also necessary for the existence
of any IR mediator, since the maximal welfare attainable by any mediator in this setting is $(\phi^1+\phi^2)/2$,
and the maximal utility any single player can attain is $\phi^1/2$.

\begin{proof} 
Suppose first that $\phi^1=\phi^2$, and so the conditional in line 2 of Mediator~\ref{alg3} is true.
In this case, each player obtains utility $\phi^1/2$ under $\Med^v_{1}$, and the mediator is thus IR. In addition, 
since by Claim~\ref{claim:no-conflict} separation is an equilibrium on every segment $S$, the mediator is
also IC. Finally, the total welfare to the players is $(\phi^1+\phi^2)/2$, which is maximal.

Next, suppose $\phi^1>\phi^2$. 
Observe that, in this case, the mediator mixes between utility $\phi^1/2$ to player $i$ and $\phi^2/2$ to player $j$,
and utility $\phi^2/2$ to player $i$ and $\phi^1/2$ to player $j$, depending on the value of $\gamma$. Regardless
of the realizations of $\gamma$, the total welfare to the players is always $(\phi^1+\phi^2)/2$. This is the
optimal welfare attainable.

Next, the expected utility of player $i$ in $\Med^v_{1}$ is
$$\frac{2v_i - \phi^2}{\phi^1-\phi^2}\cdot \frac{\phi^1}{2}+\left(1-\frac{2v_i - 
\phi^2}{\phi^1-\phi^2}\right)\frac{\phi^2}{2} = v_i.$$
Furthermore, since $v_1+v_2\leq (\phi^1+\phi^2)/2$ and total welfare is $(\phi^1+\phi^2)/2$, the utility
to player $j$ is at least $v_j$. Thus, the mediator is IR.

Finally, since by Claim~\ref{claim:no-conflict} separation is an equilibrium on every segment $S$, the mediator is IC.
\end{proof}

\subsection{Approximately Optimal Mediator for $\phi_S^1> \frac{3}{2} \cdot \phi_S^2$}\label{sec:noJCI-2}
Suppose now that  $\phi_S^1> \frac{3}{2} \cdot \phi_S^2$ for every segment $S$.
This case is more complicated since, conditional on any  segment, 
there may be a conflict between optimal strategies and welfare maximization:
\begin{claim}\label{claim:conflict}
Conditional on any segment $S$, if $\phi_S^1> \frac{3}{2} \cdot \phi_S^2$ then pooling on $g_S^1$ is 
the dominant strategy. If $\phi_S^1\leq 3 \cdot \phi_S^2$ then separating is
welfare-maximizing, whereas if $\phi_S^1> 3 \cdot \phi_S^2$ then pooling is welfare-maximizing.
\end{claim}
\begin{proof}
Again, we will compare the only two reasonable player choices
conditional on segment $S$: (i) pooling on $g_S^1$, or
(ii) separating, with one player offering $g_S^1$ and the other $g_S^2$.
As in the proof of Claim~\ref{claim:no-conflict}, pooling leads to individual utilities $\phi_S^1/3$ and so
welfare $\frac{2}{3}\cdot\phi_S^1$, whereas separating leads to individual utilities $\phi_S^1/2$ and
$\phi_S^2/2$ and so welfare $(\phi_S^1+\phi_S^2)/2$.
Then $\frac{2}{3}\cdot\phi_S^1>(\phi_S^1+\phi_S^2)/2$ if and only if $\phi_S^1> 3 \cdot \phi_S^2$,
as claimed.

In terms of optimality, however, observe that choosing $g_S^1$ leads to utility at least $\phi^1_S/3$, whereas
choosing any other good leads to utility at most $\phi^2_S/2$.  
The former is greater than the latter whenever $\phi_S^1> \frac{3}{2} \cdot \phi_S^2$, and so constitutes
a dominant action (conditional on $S$).
 \end{proof}

Our goal is to design a mediator that is approximately optimal relative to a particular benchmark. We denote
the benchmark $\OPT$, which is the maximal welfare attainable by any mediator, even disregarding the IR
and IC constraints. We will show that the mediator below, which is a variant of the mediator $\Med^v_A$ for the case
of jointly complete information, is IR, IC, and obtains welfare at least $\frac{3}{4}\cdot\OPT$. We will then show
that this approximation factor is tight.

Before describing the mediator, recall the strategy $s_j'$ from Section~\ref{sec:baseline} that,
for every $P_j(\omega)$, chooses one of the goods that is most likely to be correct. Here we will use a different
formulation
of $s_j'$. In particular, let $s_j'(\omega)\in \arg\max_{g\in G}E\left[u_j(g,g^1_\omega)|P_j(\omega)\right]$.
That is, $s_j'$ is an optimal action for player $j$ conditional on knowing only $P_j(\omega)$, in the hypothetical
situation where $i$ always plays the most-likely correct good $g^1_\omega$ in every segment. 
Also, let $\alpha_j$ and $\alpha_i$ be the expected utilities of players $j$ and $i$ in this hypothetical situation: 
$\alpha_j=E\left[u_j(s_j'(\omega),g_\omega^1)\right]$ and 
$\alpha_i=E\left[u_i(g_\omega^1,s_j'(\omega))\right]$.
Note that, under jointly complete information, the formulation
of $s_j'$ is identical to the one given in Section~\ref{sec:baseline}. Without jointly complete information, however, they may differ.

\begin{algorithm}[H]
\caption{Mediator for $\mathcal{P}=\{0,1,2\}$ without jointly complete information}\label{alg4}
\begin{algorithmic}[1]
\Procedure{$\Med^v_2$}{$V,W$} \Comment{$V=P_1(\omega)$ and $W=P_2(\omega)$ for realized $\omega\in\Omega$}
\State $S \leftarrow V\cap W$
\State $i \leftarrow \arg\max_{k\in\{1,2\}} v_k$ 
\State $j \leftarrow 3-i$
	\If{$v_i\leq \phi^1/3$}
	\Return $\left(g^1_S, g^1_S\right)$	 \Comment{Full data-sharing}
\Else
	\State $g_j \leftarrow s_j'(\omega)$
	\State Choose $\gamma \in \left[0,1\right]$ uniformly at random.
	\If{$\gamma <\frac{3\alpha_i-3v_i}{3\alpha_i-\phi^1}$}
		\State \Return $(g^1_S,g^1_S)$
	\ElsIf{$i=1$} \Return $(g^1_S, g_j)$
		\Else~\Return $(g_j, g^1_S)$
	\EndIf
\EndIf
\EndProcedure
\end{algorithmic}
\end{algorithm}


\begin{theorem}\label{thm:conflict}
Let $v_i\geq v_j$, and suppose $v_i\leq E\left[u_i(g^1_\omega, s_j'(\omega))\right]$ and 
$v_j\leq \phi^1/2-v_i$. 
Then the mediator $\Med^v_{2}$ above is IR, IC, and has welfare $W(\Med^v_{2})\geq \frac{3}{4}\cdot \OPT$. 
\end{theorem}
We note that, unlike the conditions on $v$ in Theorems~\ref{thm:without-amazon},~\ref{thm:with-amazon}, and~\ref{thm:no-conflict}, the conditions here are not necessary but only sufficient.

\begin{proof} 
Consider first the case in which $v_i\leq \phi^1/3$. Here, the utility to each player is $\phi^1/3\geq v_i\geq v_j$, 
and so the mediator is IR. Furthermore, players always pool on $g_\omega^1$ which is an equilibrium by Claim~\ref{claim:conflict}. Thus, in this case the mediator is also IC. Finally, total welfare here is $\frac{2}{3}\cdot \phi^1$. 
Below we will show that this is within a $3/4$ factor of $\OPT$.

Before turning to the approximation factor, consider the case in which $v_i>\phi^1/3$. Observe that the total welfare
in this case is at least $\phi^1/2$. If on a given segment $S(\omega)$ the players pool on $g_\omega^1$, then
welfare here is $\frac{2}{3}\cdot \phi^1_\omega > \frac{1}{2}\cdot \phi^1_\omega$. If the players separate, then
the welfare here is at least $\phi^1_\omega/2$, since player $i$ always chooses $g_\omega^1$. 
Thus, total welfare is at least $\phi^1/2$.
 
Next, the expected utility of player $i$ in this case is
$$\frac{3\alpha_i-3v_i}{3\alpha_i-\phi^1}\cdot \frac{\phi^1}{3}+
\left(1-\frac{3\alpha_i-3v_i}{3\alpha_i-\phi^1}\right)\alpha_i = v_i.$$
Furthermore, since the total welfare is at least $\phi^1/2$, player $j$ obtains welfare at least
$\phi^1/2-v_i\geq v_j$, and so the mediator is IR here as well.
In addition, since by Claim~\ref{claim:conflict} the action $g_\omega^1$ is dominant in every segment $S(\omega)$, and
in the mediator player $i$ always chooses $g_\omega^1$ and player $j$ either pools or best responds given
her information, the mediator is IC.

Finally, we turn to the approximation factor. By Claim~\ref{claim:conflict}, the maximal welfare possible on a segment 
$S$ in which $\phi_S^1\leq 3 \cdot \phi_S^2$ is $(\phi_S^1+\phi_S^2)/2$, and our mediator guarantees welfare
at least $\frac{\phi_S^1}{2}$. Hence, here we have approximation ratio at least
$$\frac{\phi_S^1/2}{(\phi_S^1+\phi_S^2)/2} \geq \frac{\phi_S^1/2}{(\phi_S^1+\phi_S^1/3)/2}=\frac{3}{4}.$$
In addition, by Claim~\ref{claim:conflict}, the maximal welfare possible on a segment 
$S$ in which $\phi_S^1> 3 \cdot \phi_S^2$ is $2\phi_S^1/3$, and our mediator guarantees welfare
at least $\frac{\phi_S^1}{2}$. Hence, here we have approximation ratio at least
$$\frac{\phi_S^1/2}{2\phi_S^1/3} =\frac{3}{4}.$$
Since a ratio of $3/4$ is achieved for each segment, it is also achieved overall, and $W(\Med^v_2)\geq \frac{3}{4}\cdot\OPT$.
\end{proof}

We now show that the approximation factor of $3/4$ is tight, by giving an example wherein there does not exist an IR mediator, not even one that is not IC,
that achieves a factor higher than $3/4$.
\begin{example}\label{ex:IR}
Fix some $\eps>0$, and suppose there are $\lceil 1/\eps\rceil+1$ segments. Within each segment $S$ there
are two correct goods, $g_S^1$ and $g_S^2$, that satisfy $\phi_S^1=\phi_S^2/\eps$. Furthermore, all the $g_S^1$ goods
are distinct, whereas for every $S$ the good $g_S^2=g^2$ is the same. Player 1 has all the 
information, and so $P_1(\omega)=S(\omega)$ for every $\omega$, whereas player 2 has none, and so $P_2(\omega)=\Omega$
for every $\omega$. Finally, $v_1=\phi^1/2$ and $v_2=0$.

In this example, maximal welfare is achieved when all players pool on $g_S^1$ for every $S$, and so 
$\OPT=\frac{2}{3}\cdot \phi^1$. However, any IR mediator must give player 1 utility $\phi^1/2$, and the only
way to achieve this is to give player 2 no additional information, so that player 1 always chooses $g_S^1$ and
player 2 always chooses $g^2$ (this is, in fact, what the mediator $\Med^v_2$ does). 
Total welfare here is then $(\phi^1 + \phi^2)/2=(1+\eps)\phi^1/2$, and so the approximation ratio is
$$\frac{(1+\eps)\phi^1/2}{2\phi^1/3} = (1+\eps)\cdot \frac{3}{4} \xrightarrow[\eps \to 0]{} \frac{3}{4}.$$
\end{example}

In the previous example, the main bottleneck to achieving an approximation factor higher than $3/4$ was the
IR constraint. In this next example we show that even without the IR constraints, there does not exist an IC mediator 
that achieves an approximation factor higher than $4/5$.
\begin{example}\label{ex:IC}
Fix some $\eps>0$, and suppose that in every segment $S$ there are  two correct goods, $g_S^1$ and $g_S^2$, that satisfy
$\phi_S^1=\left(\frac{3}{2}+\eps\right) \phi_S^2$. Furthermore, the goods in every segment are distinct from goods in
other segments. By Claim~\ref{claim:conflict}, maximal welfare here is achieved when players separate on each 
segment, and so $\OPT = (\phi^1+\phi^2)/2$. 

However, again by Claim~\ref{claim:conflict}, the dominant action for players is to pool on $g_S^1$. Note that in
any IC mediator, no player will ever choose a good $g=g_S^2$, regardless of her information, since she can always improve her payoff by deviating to the strictly better $g_S^1$ (here we use the fact that all goods are distinct).
Thus, the most any IC mediator can attain is when all players pool on the top good, leading to welfare
$\frac{2}{3}\cdot \phi^1$ and so approximation factor
$$\frac{\frac{2}{3}\cdot \phi^1}{\frac{\phi^1+\phi^2}{2}}=\frac{\frac{2}{3}\cdot \phi^1}{\phi^1+\frac{\phi^1/(3/2+\eps)}{2}} = 
\frac{2+\frac{4\eps}{3}}{\frac{5}{2}+\eps} \xrightarrow[\eps \to 0]{} \frac{4}{5}.$$
\end{example}

\subsection{Approximately Optimal Mediator for General Setting}\label{sec:noJCI-3}
Each of the mediators in the previous two sections is IR, IC, and approximately-optimal under a particular
setting of parameters, but not in the other.
Mediator $\Med^v_1$ of Mediator~\ref{alg3} is not IC when $\phi_S^1>\frac{3}{2}\cdot \phi_S^2$ for some segment $S$,
since players wish to pool but the mediator facilitates separating.
 $\Med^v_2$ of Mediator~\ref{alg4} is not IC when $\phi_S^1<\frac{3}{2}\cdot \phi_S^2$ for some segment $S$, since players wish to separate but the mediator facilitates pooling.
In this section we show that the two mediators can be combined in order to yield an IR, IC, 
and approximately optimal mediator for any parameter setting.
The main idea is straightforward: to run $\Med^v_1$ of Mediator~\ref{alg3} on segments $S$
for which $\phi_S^1\leq\frac{3}{2}\cdot \phi_S^2$, and to run $\Med^v_2$ of Mediator~\ref{alg4}
 on segments $S$ for which $\phi_S^1>\frac{3}{2}\cdot \phi_S^2$.
 
Before formalizing this mediator, we need some notation. Let $S_1=\{S:\phi_S^1\leq\frac{3}{2}\cdot \phi_S^2\}$
and $S_2=\{S:\phi_S^1>\frac{3}{2}\cdot \phi_S^2\}$. Also, for each $\ell\in\{1,2\}$ let $\phi^1_{\ell}$ be the total weight of the most-likely correct goods conditional on $S_{\ell}$, and $\phi^2_{\ell}$ be the total weight of the second most-likely correct goods conditional $S_{\ell}$. Formally, for each $k\in\{1,2\}$,
$$\phi_{\ell}^k = \sum_{S\in S_{\ell}} \phi_S^k\cdot\pr{S|S_{\ell}}.$$
Finally, modify the definitions of $s_j'$, $\alpha_j$,
and $\alpha_i$ from Section~\ref{sec:noJCI-2} as follows. First, for any $\omega$ with $S(\omega)\in S_2$ let
$s_j'(\omega)\in \arg\max_{g\in G}E\left[u_j(g,g^1_\omega)|P_j(\omega),S_2\right]$, so it is an optimal action for player $j$ conditional on the information that $\omega \in P_j(\omega)$ and $S(\omega)\in S_2$, in the hypothetical
situation where $i$ always plays the most-likely correct good $g^1_\omega$ in every segment $S(\omega)$. 
Also, let $\alpha_j$ and $\alpha_i$ be the expected utilities of players $j$ and $i$ in this hypothetical situation: 
$\alpha_j=E\left[u_j(s_j'(\omega),g_\omega^1)|S_2\right]$ and 
$\alpha_i=E\left[u_i(g_\omega^1,s_j'(\omega))|S_2\right]$.

\begin{algorithm}[H]
\caption{Mediator for $\mathcal{P}=\{0,1,2\}$ without jointly complete information}\label{alg5}
\begin{algorithmic}[1]
\Procedure{$\Med^v_3$}{$V,W$} \Comment{$V=P_1(\omega)$ and $W=P_2(\omega)$ for realized $\omega\in\Omega$}
\State $S \leftarrow V\cap W$
\If{$\phi_S^1\leq\frac{3}{2}\cdot \phi_S^2$} \textbf{run} $\Med^v_1(S)$ with $(\phi^1_1,\phi^2_1)$ replacing  
$(\phi^1,\phi^2)$
\Else~\textbf{run} $\Med^v_2(V,W)$ with $(\phi^1_2,\phi^2_2)$ replacing   $(\phi^1,\phi^2)$
\EndIf
\EndProcedure
\end{algorithmic}
\end{algorithm}

\begin{theorem}\label{thm:both}
Let $v_i\geq v_j$, and suppose that $v_1+v_2\leq (\phi^1_1+\phi^2_1)/2$, that $v_i\leq \max\left\{E\left[u_i(g^1_\omega, s_j'(\omega))|S_2\right],\frac{\phi^1_1}{2}\right\}$, and that $v_j\leq \phi^1_2/2-v_i$.
Then the mediator $\Med^v_{3}$ above is IR, IC,  and has welfare $W(\Med^v_{3})\geq \frac{3}{4}\cdot \OPT$. 
\end{theorem}

\begin{proof}
Conditional on $S_1$, the conditions on $v$ of Theorem~\ref{thm:no-conflict} are satisfied, and so $\Med^v_{1}$
is IR and IC. Conditional on $S_2$, the conditions on $v$ of Theorem~\ref{thm:conflict} are satisfied, and so $\Med^v_{2}$
is IR and IC. That is, conditional on either $S_1$ or $S_2$, the expected utilities of either player $k$ is at least 
$v_k$. Thus, $k$'s expected utility in $\Med^v_{3}$, which is a mixture of her utilities conditional on $S_1$ and on
$S_2$, must also be greater than $v_k$. This implies that $\Med^v_{3}$ is IR. 

Furthermore, since $\Med^v_{1}$ and $\Med^v_{2}$ are IC, for either player $k$ the strategy 
$\overline{s}_k$ is a best-response to
$\overline{s}_{3-k}$ conditional on either $S_1$ or $S_2$.  Thus, $\Med^v_{3}$ is also IC.

Finally, we consider the welfare guaranteed by $\Med^v_{2}$. Conditional on $S_1$, the welfare attained is
$(\phi_1^1+\phi_1^2)/2$, by Theorem~\ref{thm:no-conflict}. This is the optimal attainable by any
mediator when $\phi_S^1\leq \frac{3}{2}\cdot\phi_S^2$ for every segment $S$, by Claim~\ref{claim:no-conflict}.
Hence, conditional on $S_1$, the welfare achieved by $\Med^v_{3}$ is equal to $\OPT$ (conditional on $S_1$).

Conditional on $S_2$, the welfare attained on any segment $S$ is at least $\frac{3}{4}$ times the welfare
attainable on this segment by any mediator (by the proof of Theorem~\ref{thm:conflict}).
Hence, the welfare achieved by $\Med^v_{3}$ is at least $\frac{3}{4}\cdot \OPT$
(conditional on $S_2$).

Since conditional on either $S_1$ or $S_2$ yields welfare at least $\frac{3}{4}$ times the optimal, the total
welfare of $\Med^v_{3}$ is $W(\Med^v_{3})\geq \frac{3}{4}\cdot \OPT$.
 \end{proof}

\section{Do Players Share More Data With or Without an Amazon?}\label{sec:share-more}
In this section we consider the intriguing question of whether players optimally share more data in the presence of
an Amazon or in its absence. While it seems intuitive that players would share more data when facing
stronger outside competition, we show here that this is not necessarily the case. We also show that the reason the intuition fails
is related to the possible conflict between equilibrium and welfare maximization identified in Claim~\ref{claim:conflict}.

Because we are interested in comparing data sharing across two different environments---without and with an Amazon---we focus
on our main interpretation of the base values $v$, namely, that they correspond to the equilibrium expected utilities absent a mediator.
Of course, because the two environments differ, so will these equilibrium expected utilities. We will thus be interested in comparing the benefit of data sharing relative to firms' no-sharing equilibrium utilities $v^{noA}$ absent an Amazon with the benefit of data sharing relative to firms' no-sharing equilibrium utilities $v^A$ in the presence of an Amazon.

We begin our analysis in Section~\ref{sec:compare-jointly-complete}, where we show that, when there is jointly complete information,
the intuition that firms share more data in the presence of an Amazon than in its absence does hold. In particular, we show that if, when there is no Amazon, data sharing strictly increases welfare relative to the welfare-maximizing equilibrium sans data sharing, then optimal data sharing strictly increases welfare in the environment with an Amazon relative to any equilibrium sans data sharing.

%
%

In Section~\ref{sec:compare-no-jointly-complete} we then drop the assumption of jointly complete information, and show that this result no longer holds. In particular, we describe two settings:
\begin{enumerate}
\item A setting where full data-sharing is optimal in the presence of an Amazon, but no data-sharing is optimal in the absence of an Amazon.
\item A setting where full data-sharing is  optimal in the {\em absence} of an Amazon, but no data-sharing is optimal in the {\em presence} of an Amazon. In this setting, data sharing is strictly beneficial to firms in the absence of an Amazon but not in its presence.
\end{enumerate}

In Section~\ref{sec:compare-no-jointly-complete} we also identify and discuss the driving force behind this 
difference, namely, the possible conflict between equilibrium and welfare maximization.

\subsection{Conditions for Sharing with Jointly Complete Information}\label{sec:compare-jointly-complete}
In this section we show that, under jointly complete information, when players can strictly benefit from data sharing in the absence
of an Amazon, they can also strictly benefit from data sharing in the presence of an Amazon.

\begin{theorem}\label{thm:compare-jointly-complete}
Fix an unmediated game $\Gamma= \left(\Omega, \mathcal{P},G^{2},(P_i(\cdot))_{i\in\mathcal{P}},(u_i)_{i\in\mathcal{P}},\pi\right)$
in which there is jointly complete information and no Amazon, and let $(v_1,v_2)$ be the expected utilities in a BNE of 
$\Gamma$ for which $v_1+v_2$ is maximal. 
Also, let $\Gamma'= \left(\Omega, \mathcal{P},G^{3},(P_i(\cdot))_{i\in\mathcal{P}},(u_i)_{i\in\mathcal{P}},\pi\right)$ be the same 
unmediated game, except with an Amazon, and let $(v_1',v_2')$ be the expected utilities in some BNE of 
$\Gamma'$. Then if there exists an IR, IC mediator $\Med$ in the game without an Amazon such that $W(\Med)>v_1+v_2$, there also exists an IR, IC mediator $\Med'$ in the game with an Amazon such that $W(\Med')>v_1'+v_2'$.
\end{theorem}

The converse of Theorem~\ref{thm:compare-jointly-complete} is not true---there are examples where players strictly benefit from
data sharing in the presence of an Amazon but not in its absence. A simple example is the game described in Figure~\ref{fig:2by2} in the Introduction, with parameters $\alpha>\beta=\gamma>\delta=0$. In the absence of an Amazon, the unique equilibrium there yields maximal welfare of 1, since at least one consumer always offers the correct good. Thus, data sharing cannot strictly improve welfare.
In the presence of an Amazon, however, data sharing can be strictly beneficial, as described in the Introduction.

\begin{proof} 
Suppose towards a contradiction that there exists an IR, IC mediator $\Med$ in the game without an Amazon such that $W(\Med)>v_1+v_2$, but that for some $(v_1',v_2')$ there does not exist an IR, IC mediator $\Med'$ in the game with an Amazon such that $W(\Med')>v_1'+v_2'$. Let $E$ be the equilibrium in the game with an Amazon that yields payoffs $(v_1',v_2')$, and let $\Med_0$ be the mediator that does nothing other than recommend players their actions from the equilibrium $E$. This mediator is IR, IC, and yields welfare
$W(\Med_0)=v_1'+v_2'$. Since, by assumption, no IR, IC mediator yields higher welfare, $\Med_0$ is optimal relative to the feasible
pair $(v_1', v_2')$. Thus, by Theorem~\ref{thm:opt-mediator}, $\Med_0$ is fully revealing to at least one of the two players.
In particular, this implies that the BNE $E$ in the game with an Amazon is one in which at least one of the players always offers
the correct good $g_\omega$ to every consumer $\omega$.

We now argue that this implies that the same $E$ is an equilibrium also in the game without an Amazon. 
The player $i$ that always offers the correct good is clearly playing a dominant strategy, both in the game with an Amazon and in the game without, as any deviation can only lower her utility.
The other player $j$ is playing a best response in the game with an Amazon. Because both player $i$ and the Amazon always offer the correct good, $j$'s strategy must maximize the probability of offering the correct good conditional on knowing (only) $P_j(\omega)$. However, this strategy is also an equilibrium in the game without an Amazon, since $i$ always offers the correct good.

The equilibrium $E$ is such that 
at least one of the players always offers the correct good, and so in the game without an Amazon it yields a total welfare of 1. However, this contradicts the assumption that there exists a mediator $\Med$ that strictly improves welfare,
since no mediator can yield welfare $W(\Med)>1$.
\end{proof}

\subsection{Data-Sharing with No Jointly Complete Information}\label{sec:compare-no-jointly-complete}
The analysis of Section~\ref{sec:compare-jointly-complete} shows that when players strictly benefit from data sharing
in the absence of an Amazon, they necessarily also benefit from data sharing in the presence of an Amazon. This is consistent with the intuition that
players can gain more by sharing data when they also face outside competition. However, this result relies on the assumption
of jointly complete information. In this section we drop the assumption, and show that in that case
this intuition is not complete. In particular, while there are settings in which players benefit more from sharing data in the presence of
an Amazon, there are also settings in which players benefit from data sharing  in the absence of an Amazon but {\em not} in the presence of an Amazon.

More specifically, in Section~\ref{sec:more-sharing-with} below we describe a setting where:
\begin{itemize}
\item In the absence of an Amazon, the optimal welfare is achieved when players share no data and not when they fully share data;
\item In the presence of an Amazon, the optimal welfare is achieved when players fully share data and not when they share no data.
\end{itemize}
In contrast, in Section~\ref{sec:more-sharing-with} we describe a setting where:
\begin{itemize}
\item In the absence of an Amazon, the optimal welfare is achieved when players fully share  data and not when they share no data;
\item In the presence of an Amazon, the optimal welfare is achieved when players share no data and not when they fully share data.
\end{itemize}

The main driving force behind the distinction between the two settings is whether or not there is a conflict between equilibrium and welfare maximization, as formalized in Claim~\ref{claim:conflict}. 
The situation described in Section~\ref{sec:more-sharing-with} is one where $\phi_S^1\geq 3\cdot \phi_S^2$
for every segment $S$. In this case, under full data sharing, it is a dominant strategy for players to pool on
$g^1_S$ (both with and without an Amazon). In the presence of an Amazon, such pooling is also welfare maximizing, by Claim~\ref{claim:conflict}. In the absence of an Amazon, however, separating is always welfare maximizing. Thus, full data sharing
leads to the optimal welfare in the former case, but not in the latter.

In contrast, the situation described in Section~\ref{sec:more-sharing-without} is one where $\phi_S^1\in (\frac{3}{2}\phi_S^2,2\phi_S^2)$. In this case, under full data sharing, in the absence of an Amazon it is an equilibrium for players to separate,
whereas in the presence of an Amazon it is dominant for players to pool (again, by Claim~\ref{claim:conflict}). Furthermore, in both settings, separating is the welfare maximizing outcome. Thus, full data sharing
leads to the optimal welfare in the absence of an Amazon, but not in its presence.


%


%

\subsubsection{More Sharing With an Amazon}\label{sec:more-sharing-with}
When there is an Amazon, if for each segment $S$ it holds that $\phi_S^1\geq 3\cdot \phi_S^2$,
then maximal welfare is achieved when players pool on $\phi_S^1$ (by Claim~\ref{claim:conflict}).
This maximal welfare is can be
achieved by full data-sharing (assuming the IR constraint is satisfied): On each segment $S$, the mediator recommends $g_S^1$ to both players,
and this is optimal for them (again by Claim~\ref{claim:conflict}). In contrast, when there is no Amazon, 
full data-sharing may {\em not} be optimal (regardless of the IR constraint), as here maximal welfare is attained by separation.

In this section we construct an example in which full data-sharing is strictly suboptimal when there is
no Amazon. While the example from the introduction illustrated this, it is not quite what
we desire, since in that example the mediator with an Amazon is {\em also} an optimal mediator without.
What we would like is an example where full data-sharing is optimal only with an Amazon.
In the analysis that follows, we first ignore the IR constraint. At the end of the section we then show that, with the proper modification, the argument holds also when taking the IR constraint into account.

%

In the following construction, for simplicity we refer to different types of consumers $\omega$ by 
 the good $g_\omega$ they desire.
There are two segments, $S^1=\{g_1\}$ and $S^2=\{g_1,g_2\}$, and
players' partitions are $\Pi_1=\{S^1, S^2\}$ and $\Pi_2=\Omega$. That is,
player 1 learns whether the realized $\omega$ lies in the segment $S^1$ (and so the correct good is $g_1$) 
or $S^2$ (and so the correct good is $g_1$ or $g_2$),
whereas player 2 does not learn anything. The prior is such that both segments are equally likely,
and that $\pr{g_2|S^2}>3\pr{g_1|S^2}>0$.
That is, in segment $S^2$, the good $g_2$ is  more than three times more likely to be correct than good $g_1$.

In the presence of an Amazon, the welfare maximizing choices lead to pooling on $g_1$ in segment
$S^1$ and pooling on $g_2$ in segment $S^2$, by Claim~\ref{claim:conflict}. This can be achieved by full
data-sharing, leading to an equilibrium in which players take these pooling actions (again, by
Claim~\ref{claim:conflict}). In contrast, with no data sharing, player 2 chooses the action $g_1$ (in both segments),
leading to lower welfare due to separation on $S^2$.

Now suppose there is no Amazon. Here the welfare maximizing choices have players separating on segment $S^2$.
Without any additional information, player 2 chooses action $g_1$, leading to such optimal separation. With
data sharing, however, player 2 takes action $g_2$ on $S^2$, pooling with player 1, leading to lower welfare.
Thus, in this example, full data-sharing is optimal in the presence of an Amazon but not in its absence. Furthermore, without an Amazon, it is optimal to not share any data.

Finally, recall that we thus far neglected the IR constraint. We now modify the example so that the conclusion holds when also taking
the IR constraints into account. In particular, duplicate the construction with a different set of segments and consumers, but where players 1 and 2 have the opposite information. Formally, in addition to the two segments, $S^1=\{g_1\}$ and 
$S^2=\{g_1,g_2\}$, suppose there are also segments $\oS^1=\{\og_1\}$ and $\oS^2=\{\og_1,\og_2\}$, where the goods $g^1$ and $g^2$ are distinct from the goods $\og^1$ and $\og^2$. Furthermore, the partitions are $\Pi_1=\{S^1, S^2, \oS^1\vee \oS^2\}$ and $\Pi_2=\{S^1\vee S^2, \oS^1, \oS^2\}$. That is, if the realized segment is $S^1$ or $S^2$, then again player 1 learns the segment and player 2 learns only that the realized segment is either $S^1$ or $S^2$. Symmetrically, if the realized
segment is $\oS^1$ or $\oS^2$, then player 2 learns the segment and player 1 learns only that the realized segment is either $\oS^1$ or $\oS^2$. Finally, suppose that the realized segment is $S^1$ or $S^2$ with probability $1/2$, and $\oS^1$ or $\oS^2$ with probability $1/2$.

Because the example is symmetric, the utilities of the two players, absent data sharing, are equal. Similarly, the utilities of the two players, under full data-sharing, are also equal. Thus, whenever full data-sharing leads to strictly higher welfare than no data-sharing
(as in the presence of an Amazon above), then it also satisfies the IR constraint. And if no data-sharing is optimal in the original example
(as in the absence of an Amazon above), then it is also optimal in the modified setting.
Thus, the conclusion that full data-sharing may be optimal in the presence of an Amazon but not in its absence
holds also when we take the IR constraint into account.

\subsubsection{More Sharing Without an Amazon}\label{sec:more-sharing-without}
When there is no Amazon, the optimal welfare is achieved when, for each segment $S$, players separate---one choosing
$g_S^1$ and the other $g_S^2$---as this leads to maximal utility $g_S^1+g_S^2$ in each segment.
Furthermore, if $\phi_S^1<2\phi_S^2$ for every segment $S$, then this maximal welfare can be
achieved by full data-sharing: On each segment $S$, the mediator recommends $g_S^1$ to
one of the players, and $g_S^2$ to the other. In contrast, in the presence of an Amazon, 
if $\phi_S^1<2\phi_S^2$ for every segment $S$ then full data-sharing may {\em not} be optimal.
In this section we construct an example in which full data-sharing is strictly suboptimal with an Amazon. 
As in Section~\ref{sec:more-sharing-with} above, in the following construction, for simplicity we refer to different types of consumers $\omega$ by the good $g_\omega$ they desire. We also ignore the IR constraint. However, again as  in Section~\ref{sec:more-sharing-with} above, the conclusion holds also when we take the IR constraint into account if we consider a modified,
symmetric example in which the segments and partitions are duplicated.

Suppose there are three segments, $S^1=\{g_1\}$, $S^2=\{g_3\}$, and $S^3=\{g_1,g_2,g_3\}$, and
players' partitions are $\Pi_1=\{S^1, S^2\vee S^3\}$ and $\Pi_2=\{S^2,S^1\vee S^3\}$. That is,
player 1 learns whether the realized $\omega$ lies in the segment $S^1$, and so $g_\omega=g_1$, or whether
$\omega$ lies in one of $S^2$ or $S^3$, and so the correct good for the realized type is one of $g_1$, $g_2$, or $g_3$
(with relative probabilities to be specified shortly). Player 2 is similar, but learns whether the realized segment is $S^2$
or  one of $S^1$ or $S^3$.
The prior is such that the segments are equally likely.
Denote by $\phi^k_3=\pr{g_k|S_3}$ the probability that good $g_k$ is correct, conditional on segment $S^3$,
and suppose that $\phi^1_3 \in \left(\frac{3\cdot \phi^2_3}{2}, 2\cdot \phi^2_3\right)$ and 
that $\phi^2_3 > \phi^3_3>\phi^1_3/3$.

Suppose there is no Amazon. The optimal welfare here, $(2+\phi^1_3+\phi^2_3)/3$, 
is attained when at least one player chooses
$g_1$ conditional on $S^1$, at least one player chooses $g_3$ conditional on $S^2$, and when
one player chooses $g_1$ and another chooses $g_2$ conditional on $S^3$. Full data-sharing
leads players to learn the correct segment, and so to maximal welfare: on segments $S^1$ and $S^2$ players pool on the correct good, and on segment $S^3$
it is an equilibrium for players to separate.

In contrast, without data sharing, while player 2 will end up choosing $g_1$ on $S^3$, player 1 will end up 
choosing $g_3$ on $S^3$. This is because player 1 cannot differentiate between $S^2$ and $S^3$, and conditional
on $S^2\vee S^3$ the strategy of choosing good $g_3$ is dominant. This, however, leads to total welfare 
$(2+\phi^1_3+\phi^3_3)/3 < (2+\phi^1_3+\phi^2_3)/3$, and so full data-sharing is better than no sharing.

Now suppose there is an Amazon. When there is no data-sharing, then conditional
on segments $S^1$ or $S^2$, welfare will be $2/3$, since for both players it is a dominant strategy to offer the correct good. 
On segment $S^3$ the strategies $g_3$ and $g_1$
of players 1 and 2, respectively, will again be dominant (as in the case of no Amazon), leading to conditional welfare
$(\phi^1_3+\phi^3_3)/2$ here.

What happens will full data-sharing? Here, on segment $S^3$, players will {\em pool} on action $g_1$,
by Claim~\ref{claim:conflict} and the assumption that $\phi_3^1>\frac{3}{2}\phi_3^2$. 
Thus, welfare conditional on $S^3$ will be 
$$\frac{2}{3}\cdot \phi^1_3<\frac{\phi^1_3+\phi^3_3}{2},$$
where the inequality follows from the assumption that $\phi^3_3>\phi^1_3/3$. Since, conditional on $S^1$ or $S^2$,
welfare is the same with full sharing and with no sharing, total (unconditional) welfare is higher with no sharing
than with full sharing.

Furthermore, observe that under no data-sharing, welfare is maximal conditional on $S_1$ or $S_2$. Under $S_3$,
higher welfare could (only) be attained if players were to separate on $g_1$ and $g_2$. However, such separation can
never occur in equilibrium: Any player that is supposed to offer $g_2$ can strictly benefit by deviating to offering $g_1$,
regardless of the other player's action (by Claim~\ref{claim:conflict} and since $\phi_3^1>\frac{3}{2}\phi_3^2$). Therefore, under no data-sharing players obtain maximal welfare subject to the IC constraint being satisfied.

Thus, in this example, full data-sharing is optimal in the absence of an Amazon,  but not in its presence. In contrast, no data-sharing is optimal in the presence of an Amazon, but not in its absence. In particular, this implies that data sharing is strictly beneficial in 
the absence of an Amazon, but not in its presence.

\section{Conclusion}\label{sec:conclusion}
In this paper we proposed a simple model to study coopetitive data sharing between firms. We designed several data-sharing schemes, in the form of mediators, that are optimal in their respective domains, with and without an Amazon.

In our model, we assumed that the mediator only facilitates data sharing between firms. One natural extension
is to allow the mediator to also use monetary transfers between them. In Appendix~\ref{sec:transfers} we analyze this
extension. We show that transfers can loosen the IR constraint, implying that, under jointly complete information,
full data-sharing is often optimal. However, when there is no jointly complete information, full data-sharing may violate the IC constraint, and in this case our mediators (even without transfers) are close to optimal.

One limitation of our model is that it considers data sharing between only two firms. What happens when there
are $n>2$ small firms, and a mediator that can facilitate sharing between all of them? 
We leave a thorough analysis of optimal mediators in this setting for future work, but
 provide some preliminary results in Appendix~\ref{sec:many-players}. In particular, we analyze the case of
jointly complete information. We show that, in coopetition without an Amazon, 
a generalization of the mediator from the introductory example  to
$n$ players is optimal. In coopetition against an Amazon, however, the situation is more complicated. Nonetheless,
we provide necessary and sufficient conditions for full data-sharing to be optimal in this case.

Coopetitive data-sharing is crucial for firms' survival in current online markets, and the model and results of this paper
have only scratched the surface of what can and should be done. 
In particular, the paper leaves numerous directions
open for future research. In addition to a thorough analysis of optimal data-sharing with more than 2
players, a particularly important direction is to analyze data sharing in a variety of 
markets different from the one studied here.

%



\bibliographystyle{ims}
\bibliography{coopetition-bib}

\newpage
\begin{center}
\begin{Large}
\textbf{Appendix}
\end{Large}
\end{center}

\section{Extensions}\label{sec:extensions}

\subsection{Transfers}\label{sec:transfers}
In this section we consider an extension to the model:  the possibility of monetary transfers between players. 
With transfers, the mediator may 
require one player to pay an amount $c$, and another to receive that amount $c$, when opting in. Assume total utilities are quasi-linear in these transfers---that
is, a player's total utility is her utility from the game plus the change in her monetary position.
How does the possibility of such transfers affect the design of mediators?

A first observation is that the optimality of mediator $\Med^v_{noA}$, under jointly complete information
and without an Amazon, is not affected, since that mediator leads to a globally optimal
sum of player's utilities. Any transfers will only change the distribution of utilities (and hence
whether or not the mediator is IR), not their sum.

When players have jointly complete information and they compete against an Amazon, however, the optimal
mediator can be different. In particular, when there are transfers, then a mediator that
is fully revealing to both players is optimal whenever  $v_1+v_2\leq 2/3$. 
If both $v_i\leq 1/3$ and $v_j\leq 1/3$, then the mediator that is fully revealing to
both is optimal even without a transfer (this is what $\Med_A^v$ does). If $v_i\geq 1/3\geq v_j$, let the mediator set
a transfer $c=(v_i-1/3)$, have player $j$ transfer $c$ to player $i$, and then always recommend $g_\omega$
to both players. After the transfer this leads to utility $1/3+c$ to player $i$ and $1/3-c$ to player $j$.
Observe that $1/3+c=1/3+(v_i-1/3) = v_i$ and $1/3-c = 1/3 -(v_i-1/3) = 2/3-v_i\geq v_j$, and so the mediator is IR.
The mediator is IC since each player plays the dominant action $g_\omega$ for every $\omega$. 
And, finally, the mediator is optimal, since the sum of utilities is $2/3$, which is globally maximal.

When there is no jointly complete information the situation is slightly more complicated, and here we consider
three settings of parameters. First, if 
$\phi_S^1\leq \frac{3}{2}\cdot \phi_S^2$,
then the mediator $\Med_1^v$ of Mediator~\ref{alg3} is IR, IC, and optimal, even without transfers.
Second, if $\phi_S^1>3\cdot \phi_S^2$,
then by Claim~\ref{claim:conflict} the pooling strategy is both dominant and welfare maximizing. In this case, a mediator
like the mediator with jointly complete information and transfers above will be optimal whenever $v_1+v_2\leq \frac{2}{3}\cdot \phi^1$. If both $v_i\leq \phi^1/3$ and $v_j\leq \phi^1/3$, then the mediator that is fully revealing to
both is optimal even without a transfer (as in $\Med_2^v$). If $v_i\geq \phi^1/3\geq v_j$, let the mediator set
a transfer $c=(v_i-\phi^1/3)$, have player $j$ transfer $c$ to player $i$, and then always recommend $g_\omega$
to both players. This mediator is IR, IC, and optimal, by the same reasoning as above.

Third, if $\phi_S^1\in ( \frac{3}{2}\cdot \phi_S^2, 3\cdot \phi_S^2)$, then full data-sharing does not always satisfy the IC constraint. This is demonstrated by the example in Section~\ref{sec:more-sharing-with}, in which
the optimal IC mediator is one in which there is no data sharing (and it is strictly better than full data-sharing). Since the
example is symmetric, players' payoffs are identical, and so transfers do not alleviate the situation. Nonetheless, recall that 
mediator  $\Med^v_{3}$ (analyzed in Theorem~\ref{thm:both}) yields welfare at least $\frac{3}{4}\cdot \OPT$, where
$\OPT$ is the maximal welfare achieved by any mediator, including one that is not IR or IC. Furthermore, 
Example~\ref{ex:IC} shows that no mediator (not even one that does no satisfy IR) can in general achieve more than
$\frac{4}{5}\cdot \OPT$. This implies that even the best mediator {\em with transfers} can yield at most 
$\frac{4}{5}\cdot \OPT$. Our mediator $\Med^v_{3}$ achieves close to this---namely, $\frac{3}{4}\cdot \OPT$---{\em without} transfers.


\subsection{Many Players}\label{sec:many-players}
Throughout this paper we analyzed the case of 2 players and possibly an Amazon. In this section we discuss an extension of the model
and results to $n>2$ players. We leave a thorough analysis of optimal mediators in this setting for future work, but
here we provide some preliminary results. In particular, we consider the case in which players have jointly complete information: 
$\bigcap_{i\in\{1,\ldots,n\}} P_i(\omega) = \{\omega\}$ for every $\omega\in\Omega$, 
where $P_i(\omega)$ is the partition element of player
$i$ when consumer type $\omega$ is realized.

The main take-away from this section is that in coopetition without an Amazon, 
a generalization of the mediator from the introductory example  to
$n$ players is optimal. In coopetition against an Amazon, however, the situation is more complicated. We provide necessary and sufficient conditions
for full data-sharing to be optimal in this case.

\subsubsection{Many Players Without Amazon}
Let $E=(s_1,\ldots,s_n)$ be a BNE of the unmediated game,
where $s_i(\omega)$ is the equilibrium strategy of player $i$ when she obtains information $P_i(\omega)$.
We will design an optimal mediator for the setting in which, for every $i$, the base value $v_i$ is at most 
the utility of player $i$ in $E$. Consider the following mediator, and observe that it uses the equilibrium
$E$ in its construction (in contrast with $\Med_{noA}^v$, which only uses the base values $v$).

\begin{algorithm}[H]
\caption{Mediator for $\mathcal{P}=\{1,\ldots,n\}$ given equilibrium $E=(s_1,\ldots,s_n)$}\label{algn}
\begin{algorithmic}[1]
\Procedure{$\Med^E_{noA}$}{$V_1,\ldots,V_n$} \Comment{$V_i=P_i(\omega)$ for realized $\omega\in\Omega$ and every $i$}
\State $\omega \leftarrow \bigcap_{i\in\{1,\ldots,n\}}V_i$
\State $g_i \leftarrow \mbox{a draw from distribution }s_i(\omega)$ for every $i$
    \If{$\left(g_i\neq g_\omega~\forall i\in\mathcal{P}\right)$}
        \Return $\left(g_\omega, \ldots, g_\omega\right)$
    \Else~\Return $\left(g_1, \ldots, g_n\right)$
    \EndIf
\EndProcedure
\end{algorithmic}
\end{algorithm}

\begin{theorem}\label{thm:noA-n}
Fix a BNE $E=(s_1,\ldots,s_n)$ of the unmediated game with no Amazon, and for every $i$ let $v_i$ be at 
most the expected utility of player $i$ in $E$. Then $\Med^E_{noA}$ is IR, IC, and optimal.
\end{theorem}

In mediator $\Med^E_{noA}$, every player obtains the same utility as under $E$, plus additional utility whenever no player correctly chooses $g_\omega$. Thus, the mediator is IR. Furthermore, since under this 
mediator, for every $\omega$ at least one player chooses $g_\omega$, and so the total welfare is 1, which is maximal.
The main challenge in the proof is to show that the mediator is IC: in the mediator each player either plays her equilibrium strategy $s_i(\omega)$ or the dominant action $g_\omega$, and the proof involves showing that, conditional on
not obtaining recommendation $g_\omega$, the strategy $s_i(\omega)$ is still optimal.

\begin{proof}
Recall that $\overline{s}$ is the strategy profile in which players always follow $\Med^E_{noA}$'s recommendation.
We first show that all players receive (weakly) higher expected utilities under $\overline{s}$ than under $E$.
Consider some $\omega\in\Omega$, and some realizations $(g_1,\ldots,g_n)$ of $(s_1(\omega),\ldots,s_n(\omega))$. Consider two events:
If $g_i\neq g_\omega$ for all $i$ then without the mediator, all players obtain utility zero.
The mediator here recommends $g_\omega$ to all, leading to positive utility to all players. Conditional on this event, then, the mediator is beneficial to
all players. The other event is the complement of this first event, in which case the mediator recommends players' BNE actions. Since the mediator's recommendation
does not change players' actions relative to $E$, it also does not affect their utilities. Overall, then, all players prefer $\overline{s}$ in $\Med^E_{noA}$ to $E$, and so the mediator is IR.

Next, we show that $\overline{s}$ is a BNE of $\Gamma^{\Med^E_{noA}}$. 
Fix $\omega$ and a player $i$, and suppose all other players $j$ play $\overline{s}_{j}$. Denote this profile
of other players as $\overline{s}_{-i}(\omega)$
To simplify notation, denote by $\Med=\Med^E_{noA}$ and by $\Med(\omega)_i$ the distribution over recommendations to player $i$
in state $\omega$. 
Also fix some $g\in \supp(s_i(\omega))$, and let $\alpha=\Pr\left[(s_i(\omega) \neq g)|P_i(\omega)\cap\left(\Med(\omega)_i=g\right)\right]$.
Then when the mediator recommends action $g$, player $i$'s expected utility from following that recommendation is
\begin{align}
E&\left[u_i(g,\overline{s}_{-i}(\omega))|P_i(\omega)\cap\left(\Med(\omega)_i=g\right)\right]\nonumber \\
&= \alpha E\left[u_i(g,\overline{s}_{-i}(\omega))|P_i(\omega)\cap\left(\Med(\omega)_i=g\right) \cap (s_i(\omega) \neq g)\right]\nonumber\\
&~~~+  (1-\alpha) E\left[u_i(g,\overline{s}_{-i}(\omega))|P_i(\omega)\cap\left(\Med(\omega)_i=g\right) \cap (s_i(\omega) = g)\right]\label{eqn1}\\
&=\alpha/n+  (1-\alpha) E\left[u_i(g,\overline{s}_{-i}(\omega))|P_i(\omega)\cap\left(\Med(\omega)_i=g\right) \cap (s_i(\omega) = g)\right]\label{eqn2}\\
&=\alpha/n+  (1-\alpha) E\left[u_i(g,s_{-i}(\omega))|P_i(\omega)\cap\left(\Med(\omega)_i=g\right) \cap (s_i(\omega) = g)\right]\label{eqn3}\\
&=\alpha/n + \frac{1-\alpha}{\Pr\left[\Med(\omega)_i= g | P_i(\omega)\cap (s_i(\omega) = g)\right]} \cdot \biggl(E\left[u_i(g,s_{-i}(\omega))|P_i(\omega)\cap (s_i(\omega) = g)\right]\nonumber\\
&~~~~ - E\left[u_i(g,s_{-i}(\omega))|P_i(\omega)\cap\left(\Med(\omega)_i\neq g\right) \cap (s_i(\omega) = g)\right]\cdot\Pr\left[\Med(\omega)_i\neq g | P_i(\omega)\cap (s_i(\omega) = g)\right]\biggr)\label{eqn4}\\
&=\alpha/n + \frac{(1-\alpha)\cdot (E\left[u_i(g,s_{-i}(\omega))|P_i(\omega)\cap (s_i(\omega) = g)\right]}{\Pr\left[\Med(\omega)_i= g | P_i(\omega)\cap (s_i(\omega) = g)\right]},\label{eqn5}
\end{align}
where (\ref{eqn1}) follows from the law of total expectation, (\ref{eqn2}) follows since $s_i(\omega)\neq\Med(\omega)_i$ implies that
the mediator recommended $g_\omega$ to all players, (\ref{eqn3}) follows since $s_i(\omega)=\Med(\omega)_i$ implies that the mediator recommended the equilibrium
action to all players, (\ref{eqn4}) follows from another application of the law of total expectation, and (\ref{eqn5}) follows since  
$\Med(\omega)_i\neq g=s_i(\omega)$ implies that $g\neq g_\omega$.

Similarly, $i$'s expected utility from not following the mediator's recommendation $g$ and playing $g'\neq g$ instead is
\begin{align*}
E&\left[u_i(g',\overline{s}_{-i}(\omega))|P_i(\omega)\cap\left(\Med(\omega)_i=g\right)\right] \nonumber\\
&= \alpha E\left[u_i(g',\overline{s}_{-i}(\omega))|P_i(\omega)\cap\left(\Med(\omega)_i=g\right) \cap (s_i(\omega) \neq g)\right]\nonumber\\
&~~~+  (1-\alpha) E\left[u_i(g',\overline{s}_{-i}(\omega))|P_i(\omega)\cap\left(\Med(\omega)_i=g\right) \cap (s_i(\omega) = g)\right]\nonumber\\
&=(1-\alpha) E\left[u_i(g',\overline{s}_{-i}(\omega))|P_i(\omega)\cap\left(\Med(\omega)_i=g\right) \cap (s_i(\omega) = g)\right]\nonumber\\
&=(1-\alpha) E\left[u_i(g',s_{-i}(\omega))|P_i(\omega)\cap\left(\Med(\omega)_i=g\right) \cap (s_i(\omega) = g)\right]\nonumber\\
&=\frac{1-\alpha}{\Pr\left[\Med(\omega)_i= g | P_i(\omega)\cap (s_i(\omega) = g)\right]} \cdot \biggl(E\left[u_i(g',s_{-i}(\omega))|P_i(\omega)\cap (s_i(\omega) = g)\right]\nonumber\\
&~~~~ - E\left[u_i(g',s_{-i}(\omega))|P_i(\omega)\cap\left(\Med(\omega)_i\neq g\right) \cap (s_i(\omega) = g)\right]\cdot\Pr\left[\Med(\omega)_i\neq g | P_i(\omega)\cap (s_i(\omega) = g)\right]\biggr)\nonumber\\
&\leq\frac{(1-\alpha)\cdot (E\left[u_i(g',s_{-i}(\omega))|P_i(\omega)\cap (s_i(\omega) = g)\right]}{\Pr\left[\Med(\omega)_i= g | P_i(\omega)\cap (s_i(\omega) = g)\right]}.\label{eqn6}
\end{align*}

Since $E$ is an equilibrium, it holds that 
$$E\left[u_i(g,s_{-i}(\omega))|P_i(\omega)\cap s_i(\omega)=g)\right]\geq E\left[u_i(g',s_{-i}(\omega))|P_i(\omega)\cap s_i(\omega)=g)\right],$$
and so 
$$E\left[u_i(g,s_{-i}(\omega))|P_i(\omega)\cap\left(\Med(\omega)_i=g\right)\right] \geq E\left[u_i(g',s_{-i}(\omega))|P_i(\omega)\cap\left(\Med(\omega)_i=g\right)\right].$$

Additionally, note that if $\Med(\omega)_i=g$ for some $g\not\in\supp(s_i(\omega))$ then player $i$ is certain that $g=g_\omega$, and so following the
mediator's recommendation is optimal. Thus, in both cases, following the mediator's recommendation is optimal, and so $\overline{s}$ is a BNE.\\

Finally, we show that $\Med^E_{noA}$ is optimal with respect to $E$. In any state $\omega$, the sum of players' utilities conditional on that state is at most 1.
This maximal utility is achieved whenever at least one player plays $g_\omega$. Under $\Med^E_{noA}$, at least one player plays $g_\omega$ in {\em every}
$\omega$. Hence, the sum of players' (unconditional) expected utilities is $W(\Med^E_{noA})=1$, the maximum possible in the game.
 \end{proof}

\subsubsection{Many Players Against Amazon}
Suppose now that there are $n$ players and an Amazon, that there is jointly complete information,
and that $v$ is an arbitrary profile of base values.
\begin{claim}
The mediator $\Med$ that is fully revealing to all players $i\in\{1,\ldots,n\}$ and always recommends $g_\omega$
 is IR, IC, and optimal if and only if  $v_i\leq 1/(n+1)$ for all $i$.
\end{claim}

\begin{proof}
The strategy $\overline{s}$ in a fully revealing $\Med$ is a BNE, since each player plays $g_\omega$, which is dominant. It is IR,
since each player $i$'s expected utility in $\Med$ is $1/(n+1)\geq v_i$. Finally, it is optimal: Whenever $k$ players correctly choose $g_\omega$,
the sum of players' utilities is $k/(k+1)$. Under $\Med$, the sum of players' utilities is $n/(n+1)$, which is the most they can jointly obtain.

The reverse direction holds by the observation that if a player has base value $v_i>1/(n+1)$ then her IR constraint is violated under a fully revealing mediator.
 \end{proof}


\end{document}

\subsection{Partially-Informed Players}\label{sec:partial-info}
Our model and analysis assumed that players 1 and 2 have jointly complete information
about the consumer's type. In this section we discuss the challenges created by relaxing this assumption, and argue that a modification of
mediator $\Med^v_{noA}$ can be beneficial to players 1 and 2, albeit not necessarily optimal.
Throughout this section, suppose that the joint information of players 1 and 2 may be a non-singleton subset of $\Omega$. 
That is, there exists some $\omega\in\Omega$
for which $\abs{P_1(\omega)\cap P_2(\omega)} >1$. 

We begin with the case in which there is no Amazon, but later argue that this mediator also works when there is an Amazon.
To describe the mediator we need some notation: For a given $\omega$, let $g_\omega^1$ be the good that is 
most likely correct conditional on $P_1(\omega)\cap P_2(\omega)$, and let $g_\omega^2$ be the second most-likely good (with $g_\omega^2=g_\omega^1$ if
$P_1(\omega)\cap P_2(\omega)$ is a singleton).
Consider the mediator $\hat{\Med}^v_{noA}$, which is the same as $\Med^v_{noA}$ except for the following modification: 
Replace line 5 of Mediator~\ref{alg1} with:
$$\mbox{\textbf{if }} g_i\not\in\left(P_1(\omega)\cap P_2(\omega)\right)~\forall i\in\{1,2\}\mbox{\textbf{ then return }}(g_\omega^1,g_\omega^2)$$
(Alternatively, the mediator could return $(g_\omega^2,g_\omega^1)$ or a distribution over $(g_\omega^1,g_\omega^2)$ and $(g_\omega^2,g_\omega^1)$.
The only difference between these choices is how the surplus is split between the two players.)

If there is a state $\omega$ in which line 5 is activated, the modified mediator leads to higher utility than equilibrium $E$ with no mediator.
Furthermore, a proof similar to that of Theorem~\ref{thm:without-amazon} shows that this mediator is IC with respect to $E$. 
It is not, however, 
always optimal: In particular, it is possible that for some $\omega$ both players choose a good that {\em does} lie in $\left(P_1(\omega)\cap P_2(\omega)\right)$,
but that the probability that this is the correct good is very small, and neither is equal to $g_\omega^1$. 
Perhaps making a different recommendation, 
to a more-likely-correct good, will lead to  higher utility?

While this is possible,
the following example shows that this may lead to a mediator that is no longer a BNE. For illustration, fix a mediator that, for each
subset $P_i(\omega)\cap P_j(\omega)$, recommends to player $i$ the most-likely-correct good $g_\omega^1$.

\begin{example}
Suppose $G=\{g_1,g_2,g_3,g_4\}$, and that for some $\omega$ and $\omega'$ we have 
$P_i(\omega)=(P_i(\omega)\cap P_j(\omega'))\cup (P_i(\omega)\cap P_j(\omega))$. Denote by $A=P_i(\omega)\cap P_j(\omega')$
and $B=P_i(\omega)\cap P_j(\omega)$, and suppose that the prior on $\Omega$ is such that $\Pr[A]=\Pr[B]$. 
Finally, suppose that neither 
$A$ nor $B$ are singletons: The correct good conditional on 
$A$ is $g_1$ with probability $3/5$ and $g_2$ with probability $2/5$, and conditional on $B$ it is $g_2$ with probability $2/5$ and $g_3$ with probability $3/5$.
Finally, suppose the equilibrium $E$ is such that player $j$ chooses $g_4$ in both $P_j(\omega)$ and $P_j(\omega')$.
Conditional on obtaining information $P_i(\omega)$, the equilibrium choice for player $i$ is thus $g_2$.

Suppose players opt into using the mediator. If state $\omega$ were realized, then the mediator would learn that the state
lies in $B$, and so would recommend good $g_3$ to player $i$. However, conditional on {\em not} receiving recommendation
$g_3$, player $i$ would infer that the state must lie in $A$, and so her best action would be $g_1$. In other words, when the mediator makes a recommendation
that is equal to player $i$'s {\em equilibrium} action $g_2$ (in particular, not $g_3$), 
player $i$ still prefers to deviate to $g_1$, and so following this mediator is not a BNE. This possible deviation may have
further implications to player $j$, who may no longer wish to play her own equilibrium action $g_4$ conditional on $P_j(\omega')$, and so on.
\end{example}

One way to avoid this problem is for the mediator to recommend the better good $g_3$ with probability less than 1, and recommend the equilibrium
actions with the remaining probability. In the example, for instance, if the
mediator recommends $g_3$ with probability $1/2$, then when the mediator recommends the equilibrium action $g_2$, good $g_2$ is indeed a best response
conditional on this recommendation, and hence following the mediator is a BNE. However, the optimal way to do this across different $\omega$'s is 
not obvious, and we leave a more complete analysis of this case for future work.\\

Next, consider the case in which there is an Amazon. Just like $\Med^v_{noA}$
is IC for the case analyzed in Section~\ref{sec:with-amazon}, the mediator $\hat{\Med}^v_{noA}$ is
IC when there is an Amazon and the other players do not have jointly complete information. Similarly to above,
in addition to being IC the mediator $\hat{\Med}^v_{noA}$ improves upon players' utilities relative to their equilibrium payoffs,
but is not optimal. 

One additional wrinkle in the case with an Amazon is that we can further optimize the sum of utilities. To do so,
let $\alpha^k_\omega = \Pr[g^k_\omega|P_1(\omega)\cap P_2(\omega)]$ for each $k\in\{1,2\}$. Then replace line 5 of Mediator~\ref{alg1}
with the following:
\begin{align*}
\mbox{\textbf{if }}&g_i\not\in\left(P_1(\omega)\cap P_2(\omega)\right)~\forall i\in\{1,2\}\mbox{\textbf{ then}}\\
&\mbox{\textbf{if }}\frac{2\cdot\alpha^1_\omega}{3}\geq \frac{\alpha^1_\omega+\alpha^2_\omega}{2}\mbox{\textbf{ then return }}(g_\omega^1,g_\omega^1)\\
&\mbox{\textbf{else return }}(g_\omega^1,g_\omega^2)
\end{align*}

With this optimization, when neither of $g_1$ or $g_2$ are in the support of $P_1(\omega)\cap P_2(\omega)$, the mediator recommends
either $(g_\omega^1,g_\omega^1)$, which leads to utility $\alpha_\omega^1/3$ for each player, or $(g_\omega^1,g_\omega^2)$, which leads
to utility $\alpha_\omega^1/2$ to the first player and utility $\alpha_\omega^2/2$ to the second player. The recommendation is chosen so as
to maximize the sum of players' utilities.